\documentclass[12pt]{article}
\usepackage{amsmath}
\usepackage{amssymb}
\usepackage{amsthm}
\usepackage{fullpage}
\usepackage[utf8]{inputenc}
\usepackage{mathtools}
\usepackage{url}
\usepackage[usenames,svgnames]{xcolor}
\usepackage{cite}

\usepackage[title,titletoc,toc]{appendix}

\usepackage[pagebackref=false]{hyperref} 
\usepackage[all]{hypcap} 

\theoremstyle{plain}
\newtheorem{theorem}{Theorem}[section]
\newtheorem{lemma}[theorem]{Lemma}
\newtheorem{proposition}[theorem]{Proposition}
\newtheorem{corollary}[theorem]{Corollary}
\newtheorem{example}{Example}[section]
\newtheorem{examples}{Example}[subsection]

\newtheorem{remark}{Remark}[section]

\theoremstyle{definition}
\newtheorem{definition}{Definition}[section]

\numberwithin{equation}{section} 

\hypersetup{
breaklinks,
colorlinks,
citecolor=Green,
linkcolor=BlueViolet,
urlcolor=DarkBlue,
pdftitle={Generating functions for weighted Hurwitz numbers},
pdfauthor={Mathieu Guay-Paquet and John Harnad},
}

\DeclareMathOperator{\tr}{tr}

\DeclareMathOperator{\cyc}{cyc}

\newcommand{\id}{\mathrm{Id}}
\DeclareMathOperator{\GL}{GL}
\DeclareMathOperator{\Li}{Li}
\DeclareMathOperator{\aut}{aut}

\DeclarePairedDelimiter{\abs}{|}{|}

\DeclarePairedDelimiter{\bra}{\langle}{|}
\DeclarePairedDelimiter{\ket}{|}{\rangle}
\DeclarePairedDelimiter{\no}{:}{:}

\DeclarePairedDelimiter{\expectation}{\langle}{\rangle}

\def\be{\begin{equation}}
\def\ee{\end{equation}}

\def\bea{\begin{eqnarray}}
\def\eea{\end{eqnarray}}

\def\bt{\begin{theorem}}
\def\et{\end{theorem}}

\def\bex{\begin{example}\small \rm}
\def\eex{\end{example}}

\def\bexs{\begin{examples}\small \rm}
\def\eexs{\end{examples}}

\def\ra{\rightarrow}
\def\ss{\subset}
\def\deq{\coloneqq}

\def\br{\begin{remark}\small \rm}
\def\er{\end{remark}}

\def\&{&{\hskip -20pt}}

\def\FF{\mathcal{F}}
\def\JJ{\mathcal{J}}

\def\ZZ{\mathcal{Z}}

\def\Cb{\mathbf{C}}

\def\Nb{\mathbf{N}}
\def\Pb{\mathbf{P}}

\def\Zb{\mathbf{Z}}

\def\grF{\mathfrak{F}}

\begin{document}
\baselineskip 16pt
\medskip
\begin{center}
\begin{Large}\fontfamily{cmss}
\fontsize{17pt}{27pt}
\selectfont
\textbf{Generating functions for weighted Hurwitz numbers}\footnote{Work supported by the Natural Sciences and Engineering Research Council of Canada (NSERC) and the Fonds de recherche du Qu\'ebec -- Nature et technologies (FRQNT).}
\end{Large}\\
\bigskip
\begin{large} {Mathieu Guay-Paquet}$^{1}$ and {J. Harnad}$^{2,3}$
\end{large}
\\
\bigskip
\begin{small}
$^{1}${\em Universit\'e du Qu\'ebec \`a Montr\'eal\\
201 Av du Pr\'esident-Kennedy,
Montr\'eal, QC, Canada H2X~3Y7 \\
email: mathieu.guaypaquet@lacim.ca} \\
\smallskip
$^{2}${\em Centre de recherches math\'ematiques,
Universit\'e de Montr\'eal\\ C.~P.~6128, succ. centre ville,
Montr\'eal,
QC, Canada H3C 3J7 \\ e-mail: harnad@crm.umontreal.ca} \\
\smallskip
$^{3}${\em Department of Mathematics and
Statistics, Concordia University\\ 7141 Sherbrooke W.,
Montr\'eal, QC
Canada H4B 1R6} \\
\end{small}
\end{center}
\bigskip

\begin{abstract}
Double Hurwitz numbers enumerating weighted $n$-sheeted branched coverings of the Riemann sphere  or, equivalently, weighted paths in the Cayley graph of $S_n$ generated by transpositions, are determined by an associated weight generating function. A uniquely determined  $1$-parameter family of 2D Toda $\tau$-functions of hypergeometric type is shown to consist of  generating functions for such weighted Hurwitz numbers.  Four  classical cases are detailed, in which the weighting is uniform:  Okounkov's double  Hurwitz numbers, for which the ramification is simple at all but two specified branch points; the case of Belyi curves, with three branch points, two with specified profiles; the general case, with a specified number of branch points, two with fixed profiles, the rest  constrained  only by the genus;  and the signed enumeration case, with sign determined by the parity of the number of branch points.  Using the exponentiated quantum dilogarithm function as weight generator, three new types  of weighted enumerations are introduced. These determine  {\em quantum} Hurwitz numbers depending on a  deformation parameter $q$.  By suitable interpretation of  $q$, the  statistical mechanics of  quantum weighted  branched covers may be related to that of Bosonic gases.  The standard double Hurwitz numbers are recovered  in the classical limit.
\end{abstract}
\break 

\section{Introduction: weighted branched coverings, paths and $\tau$-functions}

\subsection{Toda \texorpdfstring{$\tau$}{tau}-functions as generating functions for Hurwitz numbers}

In \cite{GH, HO2} a method was developed for constructing parametric families of 2D Toda $\tau$-functions
\cite{Ta, UTa, Takeb} of hypergeometric type \cite{KMMM, OrS} that serve as generating functions for Hurwitz numbers, which count various configurations of branched coverings of the Riemann sphere. A natural combinatorial construction was shown to yield an equivalent interpretation in terms of path-counting in the Cayley graph of the symmetric group $S_n$ generated by transpositions. All previously known cases \cite{P, Ok, GGN1, GGN2,  Z,  BEMS,  AC1, AMMN, KZ} were placed within this framework and several new examples were deduced and explained,  both from the enumerative geometric and the combinatorial viewpoint.  (For an overview of currently known cases, including the classes of examples introduced in this work, see ref. \cite{H2}.)

In the present work, this approach is extended to  general one-parameter families of
2D Toda $\tau$-functions of hypergeometric type,  determined by an associated weight generating function
$G(z)$, together with its dual  $\tilde{G}(z)={1/G(-z)}$.  The resulting $\tau$-functions  may be interpreted as generating functions for various types of weighted enumerations of branched covers of the Riemann sphere or, equivalently, weighted paths in the Cayley graph. By suitably choosing $G(z)$, it is straightforward to recover all previously studied  cases \cite{P, Ok, GGN1, GGN2, GH, HO2, Z, BEMS, AC1, AC2, AMMN, KZ} and add an infinite variety of new ones.

In this setting, the double Hurwitz numbers  of Okounkov \cite{Ok},  which count covers that include  a pair of branch points with specified ramification profiles $\mu$ and $\nu$, together with a number $d$ of additional simple  branch points, correspond to the exponential weight generating function $G(z)=  e^z $. The associated combinatorial problem consists of counting $d$-step paths in the Cayley graph of $S_n$ starting from an element of the conjugacy class $\cyc(\nu)$ with cycle lengths given by  the partition $\nu$ and ending in the class $\cyc(\mu)$. The counting of Belyi coverings \cite{Z, KZ, AC1} of fixed genus, having again a pair of branch points with specified ramification type, plus a third whose profile length is determined by the genus, was shown in \cite{GH, HO2} to be equivalent to counting  paths in the Cayley graph that are strictly monotonically increasing in the  second elements of the successive transpositions.  This case corresponds to the weight generating function $G(z)=1+z$.  The counting of coverings with $k+2$ branch points, two of which again have specified ramification profiles $\mu$ and $\nu$, while the remaining $k$ are constrained  to provide a  specified genus $g$,  corresponds to $G(z)= (1+z)^k$. Combinatorially, this counts $d= 2g -2 +\ell(\mu) + \ell(\mu)$ step paths in the Cayley graph, starting at an element of  $\cyc(\nu)$ and ending in $\cyc(\mu)$, consisting of $k$ subsequences each of which is strictly monotonically increasing in the second elements of successive transpositions. The signed counting of branched covers, again with two ramification profiles specified, plus an arbitrary number of further branch points constrained to provide the fixed genus, was shown to correspond to counting weakly monotonic paths \cite{GGN1, GGN2, GH, HO2}. The weight generating function in this case is $G(z)=(1-z)^{-1}$, so this is the dual of the Belyi curve case.

  All these examples may be viewed as special ``classical'' cases of the more general family of weighted coverings, and associated weighted paths,
characterized by the fact that the weights are uniform (i.e., constant over the class of permissible ramification profiles), possibly up to a sign determined by the parity of the number branch points.

For a general weight generating function $G(z)$, the weights assigned  to configurations of branched coverings 
of the Riemann sphere  or to paths in the Cayley graph generated by transpositions  are determined  by the 
coefficients of the Taylor expansion  of $G(z)$ about the origin or, alternatively, by 
the coefficients in a finite or infinite product expansion of the form  $G(z) = \prod_{i=1}^\infty (1+ c_i z)$.
The algebra $\Lambda$ of symmetric  functions in an arbitrary number of indeterminates \cite{Mac}  turns out 
to be fundamental in the analysis.  

Remarkably, all six  standard bases  of $\Lambda$: $\{s_\lambda\}, \{p_\lambda\}, \{e_\lambda\}, \{h_\lambda\},
 \{m_\lambda\}$ and $\{f_\lambda\}$, labelled by integer partitions $\lambda$,  turn out to play an essential r\^ole.  
The  Schur functions $\{s_\lambda\}$, and diagonal tensor products of these, serve as expansion bases 
for the KP or 2D Toda $\tau$-functions of hypergeometric type \cite{KMMM, OrS},  in which the coefficients
are of {\it content product} form \cite{GH, HO1}.
 The coefficients in the basis  $\{p_\lambda\}$ of power sum symmetric functions, or tensor products of these,
turn out to be the  weighted Hurwitz numbers.  Provided these series  are  uniformly convergent on some
domain of the expansion parameters, they can be shown  \cite{Ta, UTa, Takeb} 
to satisfy an infinite system of Hirota bilinear equations, implying an infinite set of quadratic relations  
between the coefficients,  equivalent to the Pl\"ucker relations for an infinite dimensional Grassmann manifold.
The bases $\{e_\lambda\}$, $\{h_\lambda\}$ formed from  products of the elementary and complete 
symmetric functions, evaluated at the set of parameters $(c_1, c_2, \dots)$ appearing in the weight generating
function $G(z)$ and its dual  $\tilde{G}(z)$,  provide the path weights in the Cayley graph of $S_n$ generated by the transpositions. The monomial sum and ``forgotten'' \cite{Mac} symmetric functions   $\{m_\lambda\}$ and $\{f_\lambda\}$ similarly provide theweights of the branched coverings. The equality between the combinatorial and geometric definitions  of the weighted Hurwitz numbers follows from the various forms of the Cauchy-Littlewood
identity \cite{Mac} and its dual (see  \autoref{CL_approach}).

 A  class of examples of special interest appears when the  quantum dilogarithm function  is used to define the weight 
 generating function $G(z)$. This leads to the notion of $q$-deformed, or {\it quantum} Hurwitz numbers. In Sections \ref{examples} and \ref{quantum_hurwitz}, three variants are studied, which may be seen as $q$-deformations of the previously considered generating functions for strictly and weakly monotonic path counting. The classical limit is shown to reproduce the double Hurwitz numbers ${\rm Cov}_d(\mu, \nu)$ of \cite{Ok}. 

In the general setting, the number of branch points may be viewed as a random variable, as can the  Hurwitz numbers themselves. In the special case of quantum Hurwitz numbers, the state space is identifiable with that of a Bosonic gas with linearly spaced energy eigenvalues and fixed total energy. If the energy is taken as proportional to the degree of degeneration of the covering over the various branch points, fixing the total energy corresponds to fixing the genus of the covering curve or, equivalently, the number of steps in the Cayley graph.


\subsection{Weight generating functions and weighted Hurwitz numbers}

The use of one parameter families of 2D Toda $\tau$-functions of hypergeometric type as generating functions for  weighted branched coverings or weighted paths in the Cayley graph will be developed in detail in \autoref{parametric_families}. In the first setting, we introduce a weight generating function $G(z)$ and its dual $\tilde{G}(z)$ in the form of  infinite products
\bea
G(z) &\&:= \prod_{k=1}^\infty (1+ c_k z) \\
\tilde{G}(z) &\& := \prod_{k=1}^\infty (1- c_k z)^{-1}
\label{G_inf_prod}
\eea
Define the {\it colength} of a partition $\mu$ to be the difference between its length  and its weight 
\be
\ell^*(\mu) := |\mu|-  \ell(\mu).
\ee
The Euler characteristic of a connected $n$-sheeted branched cover of the Riemann sphere with $k+2$  branch points
with ramification profiles given by the partitions $(\mu^{(1)}, \dots, \mu^{(k)}, \mu, \nu)$ is given by
the Riemann-Hurwitz formula
\be
2- 2g= \ell(\mu) + \ell(\nu) - d, 
\label{riemann_hurwitz}
\ee
where $g$ is the genus and $d$ the sum of the colengths
\be
d  = \sum_{i=1}^k \ell^*(\mu^{(i)}).
\ee
The weights attached to a branched covering having two specified branch points with ramification
profiles of type $\mu$ and $\nu$, and $k$  additional branch points with ramification profiles 
$(\mu^{1)}, \dots, \mu^{(k)})$ are defined to be 
\bea
W_G(\mu^{(1)}, \dots, \mu^{(k)}) &\& := m_\lambda ({\bf c}) =
\frac{1}{\abs{\aut(\lambda)}} \sum_{\sigma\in S_k} \sum_{1 \le i_1 < \cdots < i_k}
 c_{i_\sigma(1)}^{\ell^*(\mu^{(1)})} \cdots c_{i_\sigma(k)}^{\ell^*(\mu^{(k)})},
\label{W_G_def}
\\\
W_{\tilde{G}}(\mu^{(1)}, \dots, \mu^{(k)}) &\& :=
f_\lambda ({\bf c})=
\frac{(-1)^{\ell^*(\lambda)}}{\abs{\aut(\lambda)}} \sum_{\sigma\in S_k} \sum_{1 \le i_1 \le \cdots \le i_k} 
c_{i_\sigma(1)}^{\ell^*(\mu^{(1)})},  \cdots c_{i_\sigma(k)}^{\ell^*(\mu^{(k)})},
\label{W_G_tilde_def}
\eea
where $m_\lambda({\bf c})$ and $f_\lambda({\bf c})$ are, respectively, the monomial sum 
 and ``forgotten'' symmetric functions \cite{Mac}  in the variables ${\bf c} = (c_1, c_2, \dots )$,
corresponding to the partition $\lambda$ of length $k$ with parts 
$(\lambda_1, \dots,\lambda_k)$ equal to  the colengths $(\ell^*(\mu^{(1)}), \dots, \ell^*(\mu^{(k)}))$,
arranged in weakly decreasing order, and $\abs{\aut(\lambda)}$ is the order of the automorphism group, under $S_n$ of the partition $\lambda$:
\be
\abs{\aut(\lambda)} := \prod_{i=1}^{\ell(\lambda} (m_i(\lambda))!,
\ee
where $m_i(\lambda)$ is the number of times $i$ appears as a part of $\lambda$.
The weighted numbers of such $n$-sheeted branched coverings of the Riemann sphere, having two specified
branch points with ramification profiles  $\mu$ and $\nu$  and any number $k$ of further
ones, with arbitrary ramification profiles $(\mu^{(1)}, \dots, \mu^{(k)})$,  are defined to be
  \bea
H^d_G(\mu, \nu) &\&\deq \sum_{k=0}^\infty \sideset{}{'}\sum_{\substack{\mu^{(1)}, \dots \mu^{(k)} \\ \sum_{i=1}^k \ell^*(\mu^{(i)})= d}}
W_{G}(\mu^{(1)}, \dots, \mu^{(k)}) H(\mu^{(1)}, \dots, \mu^{(k)}, \mu, \nu) 
\label{Hd_G}
\\
H^d_{\tilde{G}}(\mu, \nu) &\&\deq \sum_{k=0}^\infty \sideset{}{'}\sum_{\substack{\mu^{(1)}, \dots \mu^{(k)} \\ \sum_{i=1}^d \ell^*(\mu^{(i)})= d}}
W_{\tilde{G}}(\mu^{(1)}, \dots, \mu^{(k)}) H(\mu^{(1)}, \dots, \mu^{(k)}, \mu, \nu), 
\label{Hd_tildeG}
\eea
where $H(\mu^{1)}, \dots, \mu^{(k)}, \mu, \nu)$ is the number of
inequivalent  $n$-sheeted branched coverings of the Riemann sphere  (not necessarily connected) 
having $k+2$ branch points with profiles  $(\mu^{1)}, \dots, \mu^{(k)}, \mu, \nu)$
weighted by the inverse of the order of their automorphism groups, 
and $\sum'$ denotes the sum over all partitions other than the cycle type of the identity element.

For any $N\in \Zb$ and any partition $\lambda$, choosing a nonvanishing parameter $\beta$, we define the {\it content product}
\bea
r_\lambda^{(G, \beta)}(N) &\&\deq r_0^{(G, \beta)}(N) \prod_{(i,j)\in \lambda} G(\beta(N+ j-i)),
\label{r_lambda_G}
 \\
 r_\lambda^{{(\tilde{G}, \beta)}}(N) &\&\deq r^{(\tilde{G}, \beta)}_0(N) \prod_{(i,j)\in \lambda} \tilde{G}(\beta(N+ j-i)),
 \label{r_lambda_G_tilde}
 \eea
where
\bea
r_0^{(G, \beta)}(N) := \prod_{j=1}^{N-1} G((N-j)\beta)^j, \quad r_0^{(G, \beta)}(0) &\& := 1, \quad r_0^{G, \beta}(-N) := \prod_{j=1}^{N} G((j-N)\beta)^{-j},
\quad N\geq 1. \cr
&\&
\\
r^{(\tilde{G}, \beta)}_0(N): = \prod_{j=1}^{N-1}\tilde{G}((N-j)\beta)^j, \quad r^{(\tilde{G}, \beta)}_0(0) &\& = 1,\
 \quad r^{(\tilde{G}, \beta)}_0(-N) := \prod_{j=1}^{N} \tilde{G}((j-N)\beta)^{-j},
\quad N\geq 1. \cr
&\&
\label{content_product_tildeG_N}
\eea

These coefficients determine a pair of $2D$ Toda $\tau$-functions $\tau^{(G, \beta)}({\bf t}, {\bf s})$, $\tau^{(\tilde{G}, \beta)}({\bf t}, {\bf s})$ of hypergeometric type \cite{OrS, GH},  defined by their (diagonal) double Schur function expansions:
\bea
\tau^{(G, \beta)}(N, {\bf t}, {\bf s}) &\&:= \sum_{\lambda} r_\lambda^{(G, \beta)}(N) s_\lambda({\bf t}) s_\lambda({\bf s}),
\label{tau_G_double_schur}
\\
\tau^{(\tilde{G}, \beta)}(N, {\bf t}, {\bf s}) &\&:= \sum_{\lambda} r_\lambda^{(\tilde{G}, \beta)}(N) s_\lambda({\bf t}) s_\lambda({\bf s}),
\label{tau_G_tilde_double_schur}
\eea
where
\be
{\bf t} = (t_1, t_2, \dots), \quad {\bf s} = (s_1, s_2, \dots)
\ee
are the 2D Toda flow variables, which may be identified in this notation
with the power sums 
\be
t_i = \frac{p_i}{i}, \quad s_i = \frac{p'_i}{i}
\ee
in two independent sets of variables.
(See \cite{Mac} for notation and further definitions involving symmetric functions.)

The first main result  is :
 \begin{theorem}
\label{generating_function_weighted_coverings}
The functions
\be
\tau^{(G, \beta)}({\bf t}, {\bf s}) := \tau^{(G, \beta)}(0, {\bf t}, {\bf s}), \quad
\tau^{(\tilde{G}, \beta)}( {\bf t}, {\bf s}) :=\tau^{(\tilde{G}, \beta)}(0, {\bf t}, {\bf s}),
\ee
 when expanded in the basis of tensor products of pairs of  power sum symmetric functions $\{p_\mu\}$,
 are generating functions for the weighted  double Hurwitz numbers $H_G^d(\mu, \nu) $  and $H_{\tilde{G}} ^d(\mu, \nu) $ of 
$n$-sheeted branched coverings of the Riemann sphere with genus  $g$ given by \eqref{riemann_hurwitz}.
\bea
\tau^{(G, \beta)} ({\bf t}, {\bf s})
&\&= \sum_{d=0}^\infty \sum_{\substack{\mu, \nu, \\ \abs{\mu} = \abs{\nu}}} \beta^d H^d_G(\mu, \nu) p_\mu({\bf t}) p_\nu({\bf s}),
\label{tau_G_H}
\\
\tau^{(\tilde{G}, \beta)} ({\bf t}, {\bf s})
&\&= \sum_{d=0}^\infty \sum_{\substack{\mu, \nu, \\ \abs{\mu} = \abs{\nu}}} \beta^d H^d_{\tilde{G}}(\mu, \nu) p_\mu({\bf t}) p_\nu({\bf s}),
\label{tau_G_tilde_H}. 
\eea
\end{theorem}

In the combinatorial approach, the weight generating function is expressed as an infinite sum
\be
G(z) = \sum_{k=0}^\infty G_k z^k, \quad G_0 =1.
\label{generating_G}
\ee
To any $d$-step path
\be
h\in \cyc(\nu) \ra (a_1 \, b_1) h \ra \cdots \ra  (a_d \, b_d) \cdots (a_1 \, b_1) h \in \cyc(\mu)
\ee
 in the Cayley graph of $S_n$ generated by the transpositions $(a \, b)$, $a < b$,
 starting at an element  $h$ in the  conjugacy class $\cyc(\nu)$ with cycle lengths
equal to the parts of the partition $\nu$ and ending in the conjugacy class $\cyc(\mu)$, we
assign a {\it signature} $\lambda$, which is the partition of $d$ whose parts are equal to the
number of transpositions $(a_i \, b_i)$ in the sequence having the same second elements $b_i$,
and a weight 
\be
G_\lambda\deq \prod_{i\geq 1} (G_i)^{m_i} = \prod_{i=1}^{\ell(\lambda)} G_{\lambda_i}
\ee
where $m_i$ is the number of parts of $\lambda$ equal to $i$. 
If the generating function $G(z)$ may be represented in the infinite product form,
the coefficients $G_i$ are just the evaluation of the elementary symmetric functions,
defined by the generating function \cite{Mac}
\be
E(z) = \prod_{i=1}^\infty (1+z c_i) = \sum_{j=0}^\infty e_j({\bf c})z^j
\ee
The weights $G_\lambda$ are therefore the symmetric functions $e_\lambda$ formed
from the products 
 \be
 e_\lambda = \prod_{i=1}^{\ell(\lambda} e_{\lambda_i} 
 \ee
evaluated at ${\bf c}$
\be
G_\lambda = e_\lambda({\bf c}).
 \label{G_lambda_e_c}
\ee

Similarly, the weights $\tilde{G}_\lambda$ corresponding to the dual weight  generating function
$\tilde{G}(z)$ are obtained  from the symmetric functions $h_\lambda$  
constructed from products of the complete symmetric functions $\{h_i\}$, 
\be
h_\lambda =\prod_{i=1}^{\ell(\lambda)} h_{\lambda_i}
\ee
also by evaluation at ${\bf c}$
\be
\tilde{G}_\lambda = h_\lambda({\bf c}).
\ee

Denoting the number of $d$-step paths of signature $\lambda$ from $\cyc(\nu)$ to $\cyc(\mu)$ that
are weakly monotonically increasing in their second elements  as  $m_{\mu \nu}^\lambda$, we define the weighted combinatorial Hurwitz number for such paths to be
\be
\tilde{F}^d_G(\mu, \nu) \deq \frac{d!}{\abs{n}!} \sum_{\lambda, \ \abs{\lambda}=d} G_\lambda m^\lambda_{\mu \nu}
=: d!  F^d_G(\mu, \nu)
\label{Fd_G_def}
\ee

The next result shows that the $\tau$-function $\tau^{(G, \beta)} ({\bf t}, {\bf s})$
is also a generating function for the weighted numbers $\tilde{F}^d_G(\mu, \nu)$ of $d$-step paths from the conjugacy class of cycle type $\mu$ to that of type $\nu$, and hence the  numbers $H^d_G(\mu, \nu)$ and
$F^d_G(\mu, \nu)$ coincide.

\begin{theorem}
\label{generating_function_weighted_paths}
\be
\tau^{(G, \beta)} ({\bf t}, {\bf s})
= \sum_{d=0}^\infty \sum_{\substack{\mu, \nu \\ \abs{\mu} = \abs{\nu}}} \frac{\beta^d}{d!} \tilde{F}^d_G(\mu, \nu) p_\mu({\bf t}) p_\nu({\bf s})
\label{tau_F_Pp_expansion}
\ee
is the generating function for the numbers $\tilde{F}_G^d(\mu, \nu) $ of weighted $d$-step paths in the Cayley graph,
starting at an element in the conjugacy class of cycle type $\nu$ and ending at the conjugacy class of type $\mu$, with weights of all weakly monotonic paths of type $\lambda$ given by $G_\lambda$.
\end{theorem}
\noindent (The same result holds, of course, for the dual weight generating function $\tilde{G}$, if $G$ is
replaced by $\tilde{G}$ in \eqref{tau_F_Pp_expansion} and $\tilde{F}^d_G(\mu, \nu) $ by $\tilde{F}^d_{\tilde{G}}(\mu, \nu)$ ).

These together imply equality of the two different definitions of weighted Hurwitz numbers:
\begin{corollary}
\label{H_equals_F}
The geometrically defined Hurwitz numbers $H^d_{G}(\mu, \nu)$, $ H^d_{\tilde{G}}(\mu, \nu)$, enumerating
weighted  branched coverings of the  Riemann sphere with genus given by
\eqref{riemann_hurwitz}, are equal to the combinatorial Hurwitz numbers $F^d_{G}(\mu, \nu)$, $ F^d_{\tilde{G}}(\mu, \nu)$
enumerating weighted paths in the Cayley graph.
\be
H^d_{G}(\mu, \nu) = F^d_{G}(\mu, \nu), \quad H^d_{\tilde{G}}(\mu, \nu) = F^d_{\tilde{G}}(\mu, \nu).
\label{Hd_equals_Fd}
\ee
\end{corollary}

The proofs of  these results are given in Sections \ref{weighted_geometric_expansions} -   \ref{CL_approach}. 
The first follows directly from the Frobenius character formula \cite{FH, Mac},
\be
s_\lambda = \sum_{\mu, \, \abs{\mu} = \abs{\lambda}} z_\mu^{-1} \chi_\lambda(\mu) p_\mu,
\label{frobenius_character}
\ee
where $\chi_\lambda(\mu)$ is the character of the irreducible representation of symmetry
type $\lambda$ evaluated on the conjugacy class of cycle type $\mu$ and
\be
z_\mu = \prod_{i=1}^{\abs{\mu}} i^{m_i} (m_i)!, \qquad m_i = \text{number of parts of $\mu$ equal to $i$}
\ee
is the order of  the stabilizer of any element of the conjugacy class $\mu$, 
together with the Frobenius-Schur formula \cite{Frob},  \cite[Appendix~A]{LZ} expressing the Hurwitz numbers in terms of $S_n$ group characters
\be
H(\mu^{(1)}, \dots, \mu^{(k)}) =\sum_{\lambda} h(\lambda)^{k-2} \prod_{i=1}^k \frac{\chi_\lambda(\mu^{(i)})}{z_{\mu^{(i)}}}, 
\label{Frob_Schur}
\ee
where 
\be
h(\lambda) := \det\left({1\over (\lambda_i -i +j)!}\right)^{-1} 
\ee
denotes the product of hook lengths in the Young diagram associated to the partition $\lambda$.

The second is based on the use of the Jucys-Murphy elements $(\JJ_1, \dots, \JJ_n)$ \cite{Ju, Mu}, which
generate a commutative subalgebra within the group algebra $\Cb[S_n]$. When combined with the
weight generating function $G(z)$ in a multiplicative way, these provide elements $G(z, \JJ)$ of the center 
$ \ZZ(\Cb[S_n])$ whose eigenvalues  are given by the content products \eqref{content_product_tildeG_N}
that  define the coefficients in the double Schur functions expansion \eqref{tau_G_double_schur}. 
Applying these central elements to a basis of $\ZZ(\Cb[S_n])$ consisting of sums  $C_\mu$ of the
elements of the conjugacy class $\cyc(\mu)$ provides the combinatorial interpretation of the weighted Hurwitz
numbers $\tilde{F}^d_G(\mu, \nu)$ defined in \eqref{Fd_G_def}. The characteristic map, together with the 
orthogonality of the $S_n$ characters provides the identification of these as the coefficients in the expansion \eqref{tau_F_Pp_expansion}. 
An alternative, direct proof of the equalities (\ref{Hd_equals_Fd}), based upon the Cauchy-Littlewood identity for various
pairings of  dual bases for the algebra $\Lambda$ of symmetric functions is given in \autoref{CL_approach}.
In \autoref{fermionic_representation}, the usual fermionic  representation
of  2D-Toda $\tau$-functions as matrix elements of fermionic operators is recalled,
and the relevant group element expressed in terms of the weight generating function $G(z)$.      
          
\autoref{examples} deals with examples, showing how the four classical cases mentioned
above may be recovered within the general approach, and introducing three new examples in which the generating function $G(z)$ is defined in terms of the quantum dilogarithm function $\Li_2(q,z)$, leading to weighted paths involving the quantum deformation parameter $q$.

In \autoref{quantum_hurwitz}, the weighted Hurwitz numbers for the $q$-deformed cases are interpreted as 
quantum expectation values of Hurwitz numbers. Theorems \ref{hurwitz_Eq}, \ref{hurwitz_Hq} and \ref{hurwitz_Qqp} give the forms of the generating  $\tau$-functions for coverings with fixed genus, a pair of fixed branch points with specified ramification profiles 
($\mu$, $\nu$), and a variable number of additional branch points, counted either with positive weight factors, or with signed factors determined by the parity of the number of branch points. The quantum weight for any configuration of branch points may be related to the energy distribution function in a quantum Bose gas with energy spectrum linear in the integers, if the energy is viewed as proportional to the degeneracy of the covering; i.e., the sum of the colengths 
$\ell^*(\mu^{(i)})$ of the ramification profiles.  By the Riemann-Hurwitz formula, fixing the total energy is thus equivalent to fixing the genus of the covering curve.

\section{Hypergeometric \texorpdfstring{$\tau$}{tau}-functions as generating functions}
\label{parametric_families}

\subsection{\texorpdfstring{$\tau^{(G, \beta)}({\bf t}, {\bf s})$}{tau (G, \beta)} as generating function for weighted branched coverings (proof of 
\autoref*{generating_function_weighted_coverings})}
\label{weighted_geometric_expansions}

The content product formula \eqref{r_lambda_G}  may  be written as
\be
r^{(G, \beta)}_\lambda= \prod_{k=0}^\infty  (c_k \beta)^{|\lambda| } \left (\frac{1}{c_k \beta}  \right )_\lambda 
\label{content_pochhammer}
\ee
where
\be
(u)_\lambda := \prod_{i=1}^{\ell(\mu)}\prod_{j=1}^{\lambda_i}(u+j-i) 
\ee
denotes the Pochhammer symbol associated to the partition $\lambda$. Let
\be
{\bf t}(u) := (u, \frac{u}{2}, \frac{u}{3},  \dots ), \quad {\bf t}={\bf t}_\infty:= (1, 0, 0, \dots ).
\label{t_u_infty}
\ee
denote these two special values for the KP flow parameters ${\bf t} =(t_1, t_2, \dots)$
In the proof of  \autoref{generating_function_weighted_coverings}, we use   the following lemma (cf.~\cite{OrS}).
\begin{lemma}
The Pochhammer symbol may be expressed as
\be
(u)_\lambda = \frac{s_\lambda({\bf t}(u))}{s_\lambda ({\bf t}_\infty)}
= \left(1+ h(\lambda)\sideset{}{'}
\sum_{\mu, \, \abs{\mu}=\abs{\lambda}}\frac{\chi_\lambda(\mu)}{z_\mu} u^{-\ell^*(\mu)}  \right)
\label{pochhammer_frobenius}
\ee
where $\sum'_{\mu, |\mu |=|\lambda|)}$ denotes the sum over all partitions other than the cycle type of the identity element $(1)^{\abs{\lambda}}$.
\end{lemma}
\begin{proof}
This follows from  the Frobenius character formula \eqref{frobenius_character}
evaluated at the special values (\ref{t_u_infty}) and the fact that
\be
s_\lambda ({\bf t}_\infty) = h(\lambda)^{-1}.
\ee
\end{proof}

\noindent {\em Proof of \autoref{generating_function_weighted_coverings}.}
Substituting \eqref{pochhammer_frobenius} into \eqref{content_pochhammer} the content product
 formula \eqref{r_lambda_G} becomes
\bea
r^{{(G, \beta)}}_\lambda&\&= \prod_{k=0}^\infty \left(1+ h(\lambda)\sideset{}{'}\sum_{\mu, \, \abs{\mu}=\abs{\lambda}}\frac{\chi_\lambda(\mu) }{z_\mu} (\beta c_k)^{\ell^*(\mu)}\right) \\
&\&= \sum_{k=0}^\infty \ \sideset{}{'}\sum_{\substack{\mu^{(1)}, \dots, \mu^{(k)} \\ \abs{\mu^{(i)}}=\abs{\lambda}}}
\ \sum_{0 \le i_1 < \dots < i_k} ^\infty \prod_{j=1}^k \frac{h(\lambda) \chi(\mu^{(j)})}{z_{\mu^{(j)}}} \beta^{\ell^*(\mu^{(j)})}
c_{i_j}^{\ell^*(\mu^{(j)})}\\
&\&= \sum_{k=0}^\infty \ \sideset{}{'}\sum_{\substack{\mu^{(1)}, \dots \mu^{(k)} \\ \abs{\mu^{(i)}}=\abs{\lambda}}} W_{G}(\mu^{(1)}, \dots, \mu^{(k)})
\prod_{j=1}^k \frac{h(\lambda) \chi(\mu^{(j)})}{z_{\mu^{(j)}}} \beta^{\sum_{i=1}^k \ell^*(\mu^{(j)})},
\eea
where $W_g(\mu^{(1)}, \dots, \mu^{(k)})$ is as defined in \eqref{W_G_def}.
Substituting this into \eqref{tau_G_double_schur} and using the Frobenius character formula \eqref{frobenius_character}
for each of the factors $s_\lambda({\bf t}) s_\lambda({\bf s})$ gives
\be
\tau^{(G, \beta)} ({\bf t}, {\bf s}) =
\sum_{d=0}^\infty \beta^d \sum_{\substack{\mu, \nu \\ \abs{\mu}=\abs{\nu}}}
H^d_{G}(\mu, \nu) p_\mu({\bf t}) p_\nu({\bf s}),
\ee
where $ H^d_{G}(\mu, \nu)$
is the weighted Hurwitz number defined in \eqref{Hd_G}.

We proceed similarly for the dual weight generating functions $\tilde{G}$. The content product formula \eqref{r_lambda_G_tilde} for this case may be written as
\bea
r^{(\tilde{G}, \beta)}_\lambda&\&= \prod_{k=0}^\infty \left(1+ h(\lambda)\sideset{}{'}\sum_{\mu, \, \abs{\mu}=\abs{\lambda}}\frac{\chi_\lambda(\mu)}{z_\mu} (-\beta c_k)^{\ell^*(\mu)}\right)^{-1} \\
&\&= \sum_{k=0}^\infty \ \sideset{}{'}\sum_{\substack{\mu^{(1)}, \dots, \mu^{(k)} \\ \abs{\mu^{(i)}}=\abs{\lambda}}}
\ \sum_{0 \le i_1 \le \dots \le i_k} ^\infty (-1)^k\ \prod_{j=1}^k \frac{h(\lambda) \chi(\mu^{(j)})}{z_{\mu^{(j)}}} (-\beta)^{\ell^*(\mu^{(j)})}
c_{i_j\ell^*(\mu^{(j)})}\\
&\&= \sum_{k=0}^\infty (-1)^k\ \ \sideset{}{'}\sum_{\substack{\mu^{(1)}, \dots \mu^{(k)} \\ \abs{\mu^{(i)}}=\abs{\lambda}}} 
W_{\tilde{G}}(\mu^{(1)}, \dots, \mu^{(k)})
\prod_{j=1}^k \frac{h(\lambda) \chi(\mu^{(j)})}{z_{\mu^{(j)}}} (-\beta)^{\sum_{i=1}^k \ell^*(\mu^{(j)})} .
\eea
Substituting this into \eqref{tau_G_double_schur} and using the Frobenius character formula \eqref{frobenius_character}
for each of the factors $s_\lambda({\bf t}) s_\lambda({\bf s})$ gives:
\be
\tau^{(\tilde{G}, \beta)} ({\bf t}, {\bf s}) =
\sum_{d=0}^\infty \beta^d \sum_{\substack{\mu, \nu \\ \abs{\mu}= \abs{\nu}}}
H^d_{\tilde{G}}(\mu, \nu) p_\mu({\bf t}) p_\nu({\bf s}),
\ee
where
\be
H^d_{\tilde{G}}(\mu, \nu) \deq \sum_{k=0}^\infty (-1)^{k+d} \sideset{}{'}\sum_{\substack{\mu^{(1)}, \dots \mu^{(k)} \\ \sum_{i=1}^d \ell^*(\mu^{(i)})= d}}
W_{\tilde{G}}(\mu^{(1)}, \dots, \mu^{(k)}) H(\mu^{(1)}, \dots, \mu^{(k)}, \mu, \nu)
\label{Hd_G_tilde}
\ee
are the weighted, signed (quantum) Hurwitz numbers that count the number of branched coverings with genus $g$ given by \eqref{riemann_hurwitz_bis} and sum of colengths $k$, with weight $W_{\tilde{G}( z)}(\mu^{(1)}, \dots, \mu^{(k)})$ for every branched covering of type $ (\mu^{(1)}, \dots, \mu^{(k)}, \mu, \nu)$. \hfill  \qed

\subsection{Jucys-Murphy elements and content products}
\label{weight_generating}

Developing further the methods introduced in \cite{GH, HO2}, we now show how to use the weight generating
functions in the form \eqref{generating_G} to construct 2D Toda $\tau$-functions \cite{Ta, UTa, Takeb} 
that are generating functions for weighted Hurwitz numbers, counting  {\em weighted paths} in the Cayley graph of $S_n$ generated by transpositions.
We refer the reader to \cite{GH, HO2}, for further details on the combinatorial approach and additional examples.

Let $(a\,b) \in S_n$ denote the transposition interchanging the elements $a$ and $b$. The Jucys-Murphy elements \cite{Ju, Mu, DG} of the group algebra $\Cb[S_n]$ are
\be
\JJ_b\deq \sum_{a=1}^{b-1}(a\,b), \quad b=1, \dots, n,
\ee
 which generate a  commutative subalgebra.  To the weight generating function $G(z)$, we associate a $1$-parameter family of elements $G(z, \JJ)$ of the center of the group algebra $\ZZ(\Cb[S_n])$ by forming the product
\be
G(\beta, \JJ) \deq \prod_{a=1}^n G(\beta\JJ_a).
\label{Gen_JM}
\ee
Under multiplication, such elements determine endomorphisms of $\ZZ(\Cb[S_n])$ that are diagonal in the basis $\{F_\lambda\}$ of orthogonal idempotents,
\be
G(\beta, \JJ) F_\lambda = r_\lambda^{(G, \beta)} F_\lambda,
\label{GF_r_lambda}
\ee
where the parametric family of eigenvalues $ r_\lambda^{G(z)}$ are given by the content product formulae \cite{OrS, GH, HO2},
\be
r_\lambda^{(G, \beta)} \deq \prod_{(i,j)\in \lambda} G(\beta(j-i))
\label{content_product_G}
\ee
taken over the coordinates of the boxes contained in the 
Young diagram of the partition $\lambda$ of weight $\abs{\lambda}=n$.

Ref.~\cite{GH} shows how to use such elements to define parametric families of 2D Toda $\tau$-functions of the form \eqref{generating_G}  that serve as generating functions for combinatorial invariants enumerating certain paths in the Cayley graph of $S_n$ generated by all transpositions.
\br
No dependence on the lattice site $N \in \Zb$ is indicated in \eqref{tau_G_double_schur}, since in the examples considered below only $N=0$ is required. The $N$ dependence is introduced in a standard way \cite{OrS, HO1}, by replacing the factor $G(z(j-i))$ in the content product formula \eqref{content_product_G} by
$G(z(N+j-i))$, and multiplying by an overall $\lambda$-independent factor $r_0^{G(z)}(N)$.
This produces a lattice of 2D Toda $\tau$-functions $\tau^G(N, {\bf t}, {\bf s})$ which, for all the cases considered below, may be explicitly expressed in terms of $\tau^G(0, {\bf t}, {\bf s}) \eqqcolon \tau^G({\bf t}, {\bf s})$ by applying a suitable transformation of the parameters involved \cite{Ok, HO2}, and an explicit multiplicative factor depending only on $N$.
\er
Substituting the Frobenius character formula \eqref{frobenius_character}
into \eqref{tau_G_double_schur},  we obtain an equivalent expansion in terms of 
products of power sum symmetric functions and a power series in $\beta$,
\be
\tau^{(G, \beta)}({\bf t}, {\bf s}) = \sum_{d=0}^\infty \sum_{\substack{\mu, \nu \\ \abs{\mu}=\abs{\nu}}}  \beta^dF_G^d(\mu, \nu) p_\mu({\bf t}) p_\nu({\bf s}).
\ee
The coefficients $ F_G^d(\mu, \nu)$ will be interpreted in the proof of  \autoref{generating_function_weighted_paths} below as the weighted enumerations of paths in the Cayley graph starting in the conjugacy class of cycle type $\nu$ and ending in the class of type $\mu$. The geometric interpretation  will also be given, in \autoref{weighted_geometric_expansions} 
below, as weighted Hurwitz numbers enumerating  $n=\abs{\mu}=\abs{\nu}$ sheeted branched covers of the Riemann sphere.

\subsection{Weighted paths in the Cayley graph}
\label{weighted_expansions}

For any partition $\lambda = (\lambda_1 \ge \cdots \ge \lambda_{\ell(\lambda} > 0)$, let
\be
m_\lambda (\JJ) = \frac{1}{\abs{\aut(\lambda)}} \sum_{\sigma \in S_{\ell(\lambda)}} \sum_{1 \le b_1 < \cdots < b_{\ell(\lambda)} \le n}
 \JJ_{b_\sigma(1)}^{\lambda_1} \cdots \JJ_{b_{\sigma(\ell(\lambda))}}^{\lambda_{\ell(\lambda)}}
\ee
be the monomial sum symmetric function evaluated on the Jucys-Murphy elements.

\begin{lemma}
\label{generating_weighted_monomial_sums}

For any weight generating function $G(z)$, we have the following expansion for $G(z, \JJ)$:
\be
G(\beta, \JJ) = \sum_{\lambda} G_\lambda \ m_\lambda(\JJ) \beta^{\abs{\lambda}},
\label{G_m_expansion}
\ee
where
\be
G_\lambda\deq \prod_{i\geq 1} (G_i)^{m_i} = \prod_{i=1}^{\ell(\lambda)} G_{\lambda_i}
= e_\lambda({\bf c}),
\ee
with $m_i$ the number of parts of $\lambda$ equal to $i$.
\end{lemma}
\begin{proof}
\bea
G(\beta, \JJ) &\&= \prod_{a=1}^n \left(\sum_{k=0}^\infty G_k \beta^k \JJ_a^k\right) \cr
&\&=
\left(\sum_{k_1=0}^\infty G_{k_1} \beta^{k_1}\JJ_1^{k_1}\right)\cdots \left(\sum_{k_n=0}^\infty G_{k_n} \beta^{k_n}\JJ_n^{k_n}\right) \cr
&\& = \sum_{d=0}^\infty \beta^d \sum_{\lambda, \ \abs{\lambda}=d} \
\sum_{b_1, \ldots, b_{\ell(\lambda)}} \left(\prod_{i=1}^{\ell(\lambda)} G_{\lambda_i} J_{b_i}^{\lambda_i}\right) \cr
&\& \cr
&\& =\sum_{\lambda} G_\lambda \ m_\lambda(\JJ) \beta^{\abs{\lambda}}.
\eea
\end{proof}

\br
Lemma \ref{generating_weighted_monomial_sums} may be understood  as expressing
the  dual version of the Cauchy-Littlewood formula \cite{Mac}, generating a sum over 
diagonal products of the elements of the symmetric function bases  $\{e_\lambda\}$ and $\{m_\lambda\}$,
\be
\prod_{i, j}(1+ x_i y_j) = \sum_{\lambda} e_\lambda({\bf x}) m_\lambda({\bf y}),
\ee
with the indeterminates $\{x_i\}_{i=1, \dots, \infty}$ replaced by the parameters
 $\{c_i\}_{i=1, \dots, \infty}$  and the $y_j$'s by $\{\beta\JJ_j\}_{j=1, \dots n}$, since the weights
 $G_\lambda$  are, by (\ref{G_lambda_e_c}),  the evaluations of the elements $e_\lambda$
 at the parameter values ${\bf c} =(c_1, c_2, \dots )$.
 A similar evaluation of the Cauchy Littlewood formula for the dual bases $\{h_\lambda\}$ and $\{m_\lambda\}$
  \be
\prod_{i j}(1- x_i y_j)^{-1} = \sum_{\lambda }h_\lambda({\bf x}) m_\lambda({\bf y})
\ee
gives the corresponding relation for the dual weight generating function
\be
\tilde{G}(\beta, \JJ) = \sum_{\lambda} \tilde{G}_\lambda \ m_\lambda(\JJ) \beta^{\abs{\lambda}},
\label{G_tilde__m_expansion}
\ee
with
\be
\tilde{G}_\lambda = h_\lambda({\bf c}).
 \label{tilde_G_lambda_h_c}
\ee
This approach is more fully detailed in \autoref{CL_approach}.
\er

Let $\cyc(\mu) \ss S_n$ denote the conjugacy class consisting of elements with cycle lengths equal 
to the parts $\mu_i$ of the partition $\mu$.
The number of elements in $\cyc(\mu)$ is:
\be
\abs{\cyc(\mu)}= \frac{\abs{\mu}!}{z_\mu}.
\label{conj_class_order}
\ee
The transpositions are denoted $(a\,b)$,  with $a$ and $b$  distinct elements of $\{1, \dots, n\}$,
ordered by convention with $a < b$.
\begin{definition}
A \emph{$d$-step path in the Cayley graph of $S_n$ (generated by all transpositions)} is an ordered sequence
\be
(h,\ 
(a_1 \, b_1) h,\ 
(a_2 \, b_2) (a_1 \, b_1) h,\ 
\ldots,\ 
(a_d \, b_d) \cdots (a_1 \, b_1) h)
\label{eq:path-def}
\ee
of $d+1$ elements of $S_n$, where consecutive elements differ by composition on the left with a transposition $(a_i \, b_i)$.
The path is said to \emph{start} at the permutation $h$ and \emph{end} at the permutation $g \deq (a_d \, b_d) \cdots (a_1 \, b_1) h$. If $h \in \cyc(\nu)$ and $g \in \cyc(\mu)$, the path will be referred to as going \emph{from $\cyc(\nu)$ to $\cyc(\mu)$}.
\end{definition}

\begin{definition}
If the sequence $b_1, b_2, \ldots, b_d$ is  either weakly or strictly increasing, then the path is said to be \emph{weakly}
(resp. \emph{strictly}) \emph{monotonic}.
\end{definition}

\begin{definition}
The \emph{signature} of the path \eqref{eq:path-def} is the partition $\lambda$ of weight $\abs{\lambda} = d$
whose parts are equal to the number of times each particular number $b_i$ appears in the sequence 
$b_1, b_2, \ldots, b_d$, expressed in weakly decreasing order.
\end{definition}

Let $\{C_\mu\}$ denote the basis of the center $\ZZ(\Cb[S_n])$ of the group algebra $\Cb[S_n]$ consisting of the sums over the elements of $\cyc(\mu)$
\be
C_\mu = \sum_{g \in \cyc(\mu)} g.
\ee

The following result follows from a simple counting argument.
\begin{lemma}
\label{generating_weighted_paths}
Multiplication by $m_\lambda(\JJ)$ defines an endomorphism of $\ZZ(\Cb[S_n])$ which, expressed in the $\{C_\mu\}$ basis, is given by
\be
m_\lambda(\JJ) C_\mu = {1\over |\mu|!} \sum_{\nu, \, \abs{\nu}=\abs{\mu}} m^\lambda_{\mu \nu} z_\nu C_\nu,
\ee
where $m^\lambda_{\mu \nu}$ is the number of monotonic $\abs{\lambda}$-step paths in the Cayley graph of $S_n$ from 
$\cyc(\nu)$ to $\cyc(\mu)$ with signature $\lambda$.
\end{lemma}

\br
Note that in the expansion above,  we must have $\abs{\mu} = \abs{\nu}$, but there is no restriction on $\abs{\lambda}$.
\er

\br
The enumerative constants $m^\lambda_{\mu \nu}$ may be interpreted in another way, that is perhaps more natural, since it puts no restrictions on the monotonicity of the path.
By careful application of the standard braid group action on the steps of a path in the Cayley graph, one can associate a unique monotonic path with the same starting and ending points and the same signature.
By counting the number of distinct rearrangements of a sequence $b_1, b_2, \ldots, b_d$ with signature $\lambda$, it follows that the total number $\tilde{m}^\lambda_{\mu \nu} $ of  $\abs{\lambda}$-step paths in the Cayley graph from $\cyc{\mu}$ to 
$\cyc(\nu)$ with signature $\lambda$ is related to its monotonic counterpart by
\be
\tilde{m}^\lambda_{\mu \nu} \deq \frac{\abs{\lambda}!}{\prod_{i=1}^{\ell(\lambda)} \lambda_i!} m^\lambda_{\mu \nu}.
\ee
\er
Assign weight
\be
\tilde{G}_\lambda \deq \left (\prod_{i=1}^{\ell(\lambda)}\lambda_i!\right) \,G_\lambda
\label{G_tilde_lambda}
\ee
to any such path of signature $\lambda$.
Then
\be
\tilde{F}^d_G (\mu, \nu) \deq \sum_{\lambda, \ \abs{\lambda}=d} \tilde{G}_\lambda \tilde{m}^\lambda_{\mu \nu}
= d! F^d_G(\mu, \nu)
\label{Fd_G_mu_nu}
\ee
is the weighted sum over all $d$-step paths, where
\be
F^d_G(\mu, \nu) \deq \frac{1}{\abs{\nu}!} \sum_{\lambda, \ \abs{\lambda}=d} G_\lambda m^\lambda_{\mu \nu}.
\ee

It then follows from Lemmas \ref{generating_weighted_monomial_sums} and \ref{generating_weighted_paths} that
\begin{proposition}
\label{weighted_path_sums}
\be
G(z,\JJ)C_\mu = \sum_{d=0}^\infty z^d \sum_{\nu, \ \abs{\nu}=\abs{\mu}} F_G^d(\mu, \nu) z_\nu C_\nu.
\label{GC_mu}
\ee

\end{proposition}

\subsection{\texorpdfstring{$\tau^{G(z)}({\bf t}, {\bf s})$}{tau G(z)} as generating function for weighted paths (proof of \autoref*{generating_function_weighted_paths})}
\label{weighted_combinatorial_expansions}

For each choice of $G(z)$, we define the corresponding 2D Toda $\tau$-function of generalized hypergeometric type (for $N=0$) by the formal series \eqref{tau_G_double_schur}. 
It follows from general considerations \cite{Takeb, KMMM,  OrS, HO1} that this is indeed a double KP $\tau$-function and that, when extended suitably to a lattice $\tau^{G(z)} (N, {\bf t}, {\bf s})$ of such $\tau$-functions, it
satisfies the corresponding system of Hirota bilinear equations of the 2D Toda hierarchy \cite{Ta, UTa, Takeb}.

Substituting the Frobenius character formula \eqref{frobenius_character}
for each of the factors $s_\lambda({\bf t}) s_\lambda({\bf s})$ into \eqref{tau_G_double_schur}, and the corresponding 
relation between the bases $\{C_\mu\}$ and $\{F_\lambda\}$
\be
F_\lambda = h(\lambda)^{-1} \sum_{\mu, \, \abs{\mu} = \abs{\lambda}} \chi_\lambda(\mu) C_\mu
\label{F_lambda_C_mu}
\ee
into eqs.~\eqref{GC_mu} and \eqref{GF_r_lambda}, equating coefficients in the $C_\mu$ basis,
and using the orthogonality relation for the characters
\be
\sum_{\mu, \, \abs{\mu}=\abs{\lambda}} \chi_\lambda(\mu) \chi_\rho(\mu) = z_\mu \delta_{\lambda \rho},
\label{character_orthog}
\ee
we obtain the expansion
\be
\tau^{G(z)} ({\bf t}, {\bf s})
= \sum_{d=0}^\infty \sum_{\substack{\mu, \nu, \\ \abs{\mu} = \abs{\mu}}} \frac{z^d}{d!} \tilde{F}^d_G(\mu, \nu) p_\mu({\bf t}) p_\nu({\bf s}), 
\label{tau_G_FPP}
\ee
with $ \tilde{F}^d_G(\mu, \nu)$ defined as in eq.~\eqref{Fd_G_mu_nu}, proving \autoref{generating_function_weighted_paths}. 

\subsection{Direct proof of Corollary \ref*{H_equals_F}: Cauchy Littlewood identity}
\label{CL_approach}

  There is an alternative way to prove the equalities \eqref{Hd_equals_Fd},
  starting from the Cauchy-Littlewood (CL) identity \cite{Mac}
  \bea
  \prod_{a} \prod_{b} (1- x_a y_b)^{-1} = \sum_\lambda s_\lambda ({\bf x})  s_\lambda ({\bf y}) 
    \label{CL_SS}
\eea
and its dual 
 \be
  \prod_a \prod_b (1 + x_a y_b)= \sum_\lambda s_\lambda ({\bf x})  s_{\lambda' }({\bf y}),
  \label{CLSS'}
\ee
(where $\lambda'$ denotes the dual partition whose Young diagram is the transpose
of that of $\lambda$).

In these  purely formal combinatorial identities,  the variables $\{x_a\}$ and $\{y_b\}$
  may be elements of any commutative ring, the set of indices $\{a\}$ and $\{b\}$,  need not be the
  same,  and they may be finite or infinite. Moreover, the identities may be rewritten in a variety of equivalent ways, involving not
  just the Schur functions, but any pair $\{u_\lambda\}$, $\{v_\lambda\}$ of dual bases for the ring of 
  symmetric functions under the  pairing defined by the standard scalar product, in which the Schur
  functions are orthonormal:
  \be
  (u_\lambda, v_\mu) =  \delta_{\lambda \mu}.
  \ee
  In particular, we have the following alternative forms \cite{Mac} of CL 
  \bea
       \prod_{a} \prod_{b} (1- x_a y_b)^{-1} &\&= \sum_\lambda h_\lambda ({\bf x})  m_\lambda ({\bf y}) 
       \label{CLhm}
                      \\
              &\&= \sum_\lambda f_\lambda ({\bf x})  e_\lambda ({\bf y}) 
                    \label{CLfe}
\eea
              and its dual:
 \bea 
      \prod_a \prod_b (1 + x_a y_b) &\&= \sum_\lambda e_\lambda ({\bf x})  m_\lambda ({\bf y}) 
            \label{DCLem}
      \\
                  &\&= \sum_\lambda m_\lambda ({\bf x})  e_\lambda ({\bf y}) 
                        \label{DCLme}
      \eea
Evaluating the identities \eqref{DCLem}, \eqref{DCLme}, with the indeterminates $\{x_a\}$
replaced by the parameters ${\bf  c}$  defining the weights generating function $G$ and its dual $\tilde{G}$  
while the indeterminates $\{y_b\}$ are chosen as the Jucys-Murphy elements $(\JJ_1, \dots, \JJ_n)$, and using \eqref{Gen_JM} gives
\bea
G( \beta,\JJ) &\&= \sum_{\lambda} e_\lambda({\bf c})  m_\lambda(\JJ)  \beta^{|\lambda|}  
\label{G_em}
\\
&\&=\sum_{\lambda}  m_\lambda ({\bf c}) e_\lambda (\JJ) \beta^{|\lambda|} .
\label{G_me}
\eea
Similarly evaluating \eqref{CLhm}, \eqref{CLfe} gives
\bea
\tilde{G}( \beta, \JJ) &\&= \sum_{\lambda} h_\lambda({\bf c}) m_\lambda(\JJ)  \beta^{|\lambda|} ,
\label{tildeG_mz}
\\
&\&=\sum_{\lambda}  f_\lambda ({\bf c}) e_\lambda (\JJ) \beta^{|\lambda|} .
\label{tildeG_fe}
\eea

The combinatorial definition of the  Hurwitz numbers is equivalent  to the formula  
\be
\frac{1}{n!} [\id] \prod_{i=1}^k C_{\mu^{(i)}} = H(\mu^{(1)}, \dots, \mu^{(k)}), 
\label{combin_hurwitz_id}
\ee
where $[\id] $ means taking the component of the identity element within the cycle sum basis $\{C_\nu\}$ of $\ZZ(\Cb[S_n])$, and $n = \abs{\mu^{(i)}}$ is the number of sheets in the enumerated covers.
More generally, we have the following  expression for the product
$\prod_{i=1}^k C_{\mu^{(i)}}$ in the $\{C_\nu \}$ basis 
\be
\prod_{i=1}^k C_{\mu^{(i)}} =  \sum_{\nu, \, |\mu|= |\mu^{(i)}} H(\mu^{(1)}, \dots , \mu^{(k)}, \nu) z_\nu C_{\nu}.
\label{combin_hurwitz_sum}
\ee
This  is equivalent to the Frobenius-Schur formula eq.~(\ref{Frob_Schur}),  as can be seen by substituting the inverse 
\be
C_{\mu^{(i)}} = z_{\mu^{(i)}}^{-1}\sum_{\lambda^{(i)}, \, |\lambda^{(i)}|=|\mu^{(i)}|} h_\lambda \chi_\lambda(\mu^{((i)}) F_\lambda
\ee
of the basis change relation (\ref{F_lambda_C_mu}) into (\ref{combin_hurwitz_sum}), using the orthogonal idempotency relations
\be
F_{\lambda^{(i)}} F_{\lambda^{(j)}} = \delta_{\lambda^{(i)} \lambda^{(j)}} F_{\lambda^{(i)}}
\ee
and reverting to the $\{C_\nu\}$ basis.

Similarly to \autoref{generating_weighted_paths}, we have the following expression for $e_\lambda(\JJ) C_\mu $
\begin{lemma}
\label{generating_weighted_covers}
Multiplication of the elements of the basis $\{C_\mu\}$  by $e_\lambda(\JJ)$  gives
\be
e_\lambda(\JJ) C_\mu = \sum_{\substack {\mu^{(1)}, \dots , \, \mu^{(k)} , \\ \{\ell^*(\mu^{(i}) = \lambda_i\}} }
\prod_{i=1}^k C_{\mu^{(i)}} C_\mu.
\label{e_lambda_C_mu}
\ee
\end{lemma}
\begin{proof}
As shown in \cite{Ju}, the elementary symmetric polynomials $e_k(\JJ)$ in the Jucys-Murphy elements 
$(\JJ_1, \dots, \JJ_n)$  are equal to the sum over the elements of $S_n$ consisting of  products of $n-k$ cycles   
\be
e_k(\JJ)  = \sum_{\substack{\mu, \\ \ell^*(\mu) = k}} C_\mu .
\ee
Substitution in 
\be
e_\lambda(\JJ)  =  \prod_{i=1}^k  e_{\lambda_i}(\JJ)
\ee
gives  (\ref{e_lambda_C_mu}).
\end{proof}

It follows from eqs.~\eqref{W_G_def}, \eqref{Hd_G}, \eqref{G_lambda_e_c},  \eqref{Fd_G_def},
\eqref{combin_hurwitz_sum} and Lemmas \ref{generating_weighted_paths} and \ref{generating_weighted_covers} that
\bea 
F^d_G(\mu, \nu) &\&= [\beta^d C_{\nu}] \sum_{\lambda} e_\lambda({\bf c})  m_\lambda(\JJ)  \beta^{|\lambda|} C_\mu, \\
H^d_G(\mu, \nu) &\&= [\beta^d C_{\nu}]  \sum_{\lambda}  m_\lambda ({\bf c}) e_\lambda (\JJ) \beta^{|\lambda|}C_\mu,
\label{F_d_H_d}
\eea
and hence, by \eqref{G_em},  \eqref{G_me}, they are equal. Similarly applying \eqref{tildeG_mz}  and \eqref{tildeG_fe} proves the dual case from 
\bea
 F^d_{\tilde{G}}(\mu, \nu) &\&= [\beta^d C_{\nu}] \sum_{\lambda} h_\lambda({\bf c})  m_\lambda(\JJ)  \beta^{|\lambda|} C_\mu,   \\
H^d_{\tilde{G}}(\mu, \nu) &\&= [\beta^d C_{\nu}]  \sum_{\lambda}  f_\lambda ({\bf c}) e_\lambda (\JJ) \beta^{|\lambda|}C_\mu .
\label{tilde_F_d_H_d}
\eea

\subsection{Fermionic representation}
\label{fermionic_representation}

Double KP $\tau$-functions of the form \eqref{tau_G_double_schur} also have a Fermionic representation
\cite{OrS, GH, HO1, HO2}
\be
\tau^G({\bf t}, {\bf s}) = \bra{0} \hat{\gamma}_+({\bf t}) \hat{C}_G(\beta) \hat{\gamma}_-({\bf s})\ket{0}
\ee
where the Fermionic operator $ \hat{C}_G $, $ \hat{\gamma}_+({\bf t}) $ and $ \hat{\gamma}_-({\bf s}) $ are
defined by
\bea
\hat{C}_G &\&= e^{\sum_{j=-\infty}^\infty T^G_j (\beta) \no{\psi_j \psi_j^\dag}} \\
\hat{\gamma}_+({\bf t}) &\& = e^{\sum_{i=1}^\infty t_i J_i}, \quad \hat{\gamma}_-({\bf s}) = e^{\sum_{i=1}^\infty s_i J_i},
\quad J_i = \sum_{k\in \Zb} \psi_k \psi^\dag_{k+i}, \quad i \in \Zb,
\eea
in terms of the Fermionic creation and annihilation operators $\{\psi_i, \psi_i^\dag\}_{i\in \Zb}$,
acting on the Fermionic Fock space,
\be
\FF = \bigoplus_{N \in \Zb} \FF_N \quad (N= \text{vacuum charge}),
\ee
satisfying the usual anticommutation relations
\be
[\psi_i, \, \psi_j^\dag]_+ = \delta_{ij}
\ee
and vacuum state $\ket{0}$ vanishing conditions
\be
\psi_i \ket{0} = 0, \quad \text{for $i < 0$}, \quad \psi^\dag_i \ket{0} = 0, \quad \text{for $i\ge 0$},
\ee
and the parameters $T^G_j(\beta)$ are defined by
\be
T^G_j(\beta) = \sum_{k=1}^j \ln G(\beta k), \quad T^G_0(\beta) = 0, \quad T^G_{-j}(\beta ) = -\sum_{k=0}^{j-1} \ln G(-\beta k) \quad\text{for $j>0$}.
\ee

This follows from the fact that $ \hat{C}_G $ is diagonal in the basis $\{\ket{\lambda; N}\}$
\be
\hat{C}_G(\beta) \ket{\lambda; N} = r_\lambda^{(G, \beta)}(N) \ket{\lambda; N}
\label{CG_lambda}
\ee
with eigenvalues
\bea
r_\lambda^{(G, \beta)}(N) &\&\deq r^{(G,\beta)}_0(N) \prod_{(i,j)\in \lambda} G(\beta(N+ j-i)), \\
r^{(G, \beta)}_0(N) = \prod_{j=1}^{N-1} G((N-j)\beta)^j, \quad r_0(0) &\& = 1, \quad r^{(G, \beta)}_0(-N) = \prod_{j=1}^{N} G((j-N)\beta)^{-j},
\quad N>1.\cr
&\&
\eea
Eq.~\eqref{CG_lambda} means that  the map
\be
\begin{split}
\grF \colon \bigoplus_{n \geq 0} \ZZ(\Cb[S_n]) &\to \FF_0 \\
\grF: F_\lambda &\mapsto \frac{1}{h_\lambda} \ket{\lambda; 0}
\end{split}
\ee
 intertwines the action of the abelian group of elements of the form $\hat{C}_G$ on $\FF_0$ with the action of the group of elements $G(z, \JJ) \in \ZZ(C(S_n])$ by multiplication on the direct sum of the centers $\ZZ(C[S_n])$ of the $S_n$ group algebras \cite{GH}.

More generally, using the charge $N$ vacuum state
\be
\ket{N} = \psi_{N-1} \cdots \psi_0 \ket{0}, \quad \ket{-N} = \psi^\dag_{-N} \cdots \psi^\dag_{-1} \ket{0},
\quad N\in \Nb^+,
\ee
we may define a 2D Toda lattice of $\tau$-functions by
\bea
\tau^{G, \beta)}(N, {\bf t}, {\bf s}) &\& = \bra{N} \hat{\gamma}_+({\bf t}) \hat{C}_G (\beta)\hat{\gamma}_-({\bf s})\ket{N} \\
&\& \cr
&\& = \sum_{\lambda} r_\lambda^{G}(N)s_\lambda({\bf t}) s_\lambda({\bf s}).
\eea
These satisfy the infinite set of Hirota bilinear equations for the 2D Toda lattice hierarchy \cite{UTa, Ta, Takeb}.

\section{Examples: classical and quantum}
\label{examples}
\subsection{Classical counting of paths: combinatorial Hurwitz numbers}
\label{classical_hurwitz}
The following four examples were studied in ref.~\cite{Ok, GGN1, GGN2, GH, HO2}.
The interpretation of the associated $\tau$-functions as generating functions for weighted enumeration of paths and as Hurwitz numbers for various types of branched covers of $\Cb\Pb^1$ of fixed genus will be recalled in what follows.

\bex {\bf Simple double Hurwitz numbers \cite{Ok}.}
\label{ex1}
This case is Okounkov's simple double Hurwitz numbers \cite{Ok}, which enumerate $n = \abs{\mu}= \abs{\nu}$ sheeted branched coverings of the Riemann sphere, with ramification types $\mu$ and $\nu$ at $0$ and $\infty$, $d$ additional simple branch points and genus $g$ given by the Riemann-Hurwitz formula
\be
2- 2g = \ell(\mu) + \ell(\nu) - d.
\label{riemann_hurwitz_bis}
\ee
The weight generating function $G(z)$ in this case is just the exponential function. The parameters entering into the associated $\tau$-function are as follows.
\bea
G(z) &\& = \exp(z)\deq e^z = \sum_{i=1}^\infty \frac{z^i}{i!}, \quad
\exp_i = \frac{1}{i!}, \quad \exp_\lambda = \frac{1}{\prod_{i=1}^{\ell(\lambda)} \lambda_i !} \\
\exp(\beta, \JJ)&\& = e^{\beta \sum_{b=1}^n \JJ_b} = \sum_{k=0}^\infty \frac{\beta^d}{d!} \left(\sum_{b=1}^n \JJ_b\right)^d \\
&\& = \sum_{d=0}^\infty \beta^d \sum_{\lambda, \, \abs{\lambda} =d}\left(\prod_{i=1}^{\ell(\lambda)}(\lambda_i)! \right)^{-1}
m_\lambda(\JJ) \\
r^{(\exp, \beta)}_j &\& = e^{j \beta}, \quad r_\lambda^{\exp} (\beta)= e^{\frac{\beta}{2} \sum_{i=1}^{\ell(\lambda)}\lambda_i (\lambda_i - 2i +1)},
\quad T^{\exp}_j = \frac{1}{2} j(j+1) \beta.
\eea
\br
This case is, in a sense, exceptional because the infinite product form \eqref{G_inf_prod} of the generating 
functions must be interpreted as a limit
\be
e^z =\lim_{k \ra \infty} (1+ \frac{1}{k} z)^k
\ee
and the sums \eqref{tau_G_double_schur}, \eqref{tau_G_H} understood in the weak sense (i.e., with the limit taken under the summation signs).
\er

For this case
\be
\tilde{F}_{\exp}^d(\mu, \nu) = d !\sum_{\lambda, \ \abs{\lambda}=d} \left(\prod_{j=1}^{\ell(\lambda)} (\lambda !)\right)^{-1}m^\lambda_{\mu \nu}
= \sum_{\lambda, \ \abs{\lambda}=d} \tilde{m}^\lambda_{\mu \nu}
\ee
is the total number of $(d+1)$-term products $(a_1\, b_1) \cdots (a_d\, b_d) g$ such that $g \in \cyc(\mu)$ and the product $(a_1\, b_1) \cdots (a_d\, b_d) g \in \cyc(\nu)$; i.e., the number of (unordered) sequences of $d$ transpositions leading from the class of type $\mu$ to the class of type $\nu$. Equivalently, it may be viewed as the number of $d$-step paths in the Cayley graph of $S_n$ generated by all transpositions, from the conjugacy class of cycle type $\mu$ to the class of cycle type $\nu$.
\eex


\bex{\bf Monotonic double Hurwitz numbers \cite{GH, HO2}.}
\label{ex2}
\bea
G(z) &\&= E(z)\deq 1+ z, \quad
E_i =\delta_{1 i}, \text{for } i\ge1, \quad e_\lambda = \delta_{\lambda, (1^{\abs{\lambda}})} \\
E(\beta, \JJ) &\&= \prod_{a=1}^n (1+\beta \JJ_a), \\
r^{(E,\beta)}_j&\& = 1 + \beta j, \quad r^{(E, \beta)}_\lambda = \prod_{(i,j)\in \lambda} (1 + \beta (j-i)) = \beta^{\abs{\lambda}} \, (1/\beta)_{\lambda} \\
T^E_j(\beta) &\& = \sum_{i=1}^j \ln(1+i \beta), \quad T^E_{-j}(\beta) = -\sum_{i=1}^{j-1}\ln(1-i \beta), \quad j > 0,
\eea
where
\be
(u)_\lambda := \prod_{i=1}^{\ell(\lambda)}(u-i+1)_{\lambda_i}
\ee
is the multiple Pochhammer symbol corresponding to the partition $\lambda$.

In this case we have
\be
\sum_{\lambda, \abs{\lambda}=d} e_\lambda m_\lambda(\JJ) = \sum_{b_1 < \dots< b_d} \JJ_{b_1} \cdots \JJ_{b_d}
\ee
and the coefficient $ F_{E}^d(\mu, \nu)$ is
\be
F_{E}^d(\mu, \nu) = m^{(1)^{d}}_{\mu \nu},
\label{F_Ed}
\ee
which enumerates all $d$-step paths in the Cayley graph of $S_n$ starting at an element in the conjugacy class of cycle type $\nu$ and ending in the class of type $\mu$, that are strictly monotonically increasing in their second elements \cite{GH}.

Equivalently \cite{HO2}, this equals the double Hurwitz numbers for Belyi curves,
\cite{Z, KZ, AC1}, which enumerate $n$-sheeted branched coverings of the Riemann sphere having three ramification points, with ramification profile types $\mu$ and $\nu$ at $0$ and $\infty$, and a single additional branch point, with ramification profile $\mu^{(1)}$ having colength
\be
\ell^*(\mu^{(1)}) \deq n - \ell(\mu^{(1)}) = d
\ee
i.e., with $n-d$ preimages. The genus is again given by the Riemann-Hurwitz formula \eqref{riemann_hurwitz_bis}.
\eex


\bex{\bf Multimonotonic double Hurwitz numbers \cite{HO2}.}
\label{ex3}
\bea
G(z) &\&=E(z)^k\deq (1+ z)^k, \quad E^k_i = \binom{k}{i}, \quad
e_\lambda^k = \prod_{i=1}^{\ell(\lambda)}
\binom{k}{\lambda_i} \\
E(\beta, \JJ)^k &\& = \prod_{a=1}^n (1+\beta \JJ_a)^k, \\
r^{(E^k, \beta)}_j&\&= (1 +  \beta j)^k, \quad r_\lambda^{(E^k, \beta)}
= \prod_{(i,j)\in \lambda} (1 + \beta (j-i))^k = \beta ^{k \abs{\lambda}} ((1/\beta)_{\lambda})^{k}, \\
T^{E^k}_j (\beta)&\& = k \sum_{i=1}^j \ln(1+iz), \quad T^{E^k}_{-j}(\beta) = - k\sum_{i=1}^{j -1}\ln(1-i \beta), \quad j > 0.
\eea

\bea
\sum_{\lambda, \abs{\lambda}=d} E^k_\lambda m_\lambda(\JJ) &\&= \sum_{\lambda, \abs{\lambda}=d} \left(\prod_{i=1}^{\ell(\lambda)} \binom{k}{\lambda_i}\right) m_\lambda(\JJ) \cr
&\& = [z^d] \prod_{a=1}^n(1+z J_a)^k
\eea
where $[z^d]$ means the coefficient of $z^d$ in the polynomial.

The coefficient
\be
F_{E^k}^d(\mu, \nu) = \sum_{\lambda, \abs{\lambda}=k} \left(\prod_{i=1}^{\ell(\lambda)} \binom{k}{\lambda_i}\right) m^\lambda_{\mu \nu}
\label{F_Ed}
\ee
is the number of $(d+1)$-term products $(a_1\, b_1) \cdots (a_d\, b_d) g$ such that $g\in \cyc(\mu)$, while the product $(a_1\, b_1) \cdots (a_d\, b_d) g \in \cyc(\nu)$, and which consist of a product of $k$ consecutive subsequences, each of which is strictly monotonically increasing in the second elements of each $(a_i\, b_i)$ \cite{GH, HO2}.

Equivalently, they are double Hurwitz numbers that enumerate $n$-sheeted branched coverings of the Riemann sphere with ramification profile types $\mu$ and $\nu$ at $0$ and $\infty$, and (at most) $k$ additional branch points,
such that the sum of the colengths of their ramification profile types
(i.e., the ``defect" in the Riemann Hurwitz formula \eqref{riemann_hurwitz_bis}) is equal to $d$:
\be
\sum_{i=1}^k \ell^*(\mu^{(i)}) = kn - \sum_{i=1}^k \ell(\mu^{(i)}) = d.
\ee
This amounts to counting covers with the genus fixed by \eqref{riemann_hurwitz_bis}
and the number of additional branch points fixed at $k$, but no restriction on their simplicity.

\eex


\bex{\bf Weakly monotonic double Hurwitz numbers \cite{GGN1, GGN2, GH}.}
\label{ex4}
 This is the dual $\tilde{E}$ of the weight generating function of example \ref{ex2}.
\bea
G(z) &\&= H(z)\deq \frac{1}{1- z} =\tilde{E}(z), \quad H_i = 1 \text{ for } i\ge1, \quad h_\lambda =1 \quad  \forall \lambda \\
H(\beta, \JJ) &\&= \prod_{a=1}^n (1-\beta \JJ_a)^{-1}, \\
r^{(H, \beta)}_j&\& = (1 - \beta j )^{-1}, \quad r_\lambda^H (\beta) =
\prod_{(i,j)\in \lambda} (1 - \beta (j-i))^{-1} = (-\beta)^{-\abs{\lambda}}((-1/\beta)_\lambda)^{-1}, \\
T^{H}_j (\beta)&\& = - \sum_{i=1}^j \ln(1-i \beta), \quad T^E_{-j}(\beta) = \sum_{i=1}^{j -1}\ln(1+i \beta), \quad j > 0.
\eea

We now have
\be
\sum_{\lambda, \abs{\lambda}=d} G_\lambda m_\lambda(\JJ) = \sum_{b_1 \le \dots \le b_d} \JJ_{b_1} \cdots \JJ_{b_d}.
\ee
and
\be
F_{H}^d(\mu, \nu) = \sum_{\lambda, \ \abs{\lambda}=k} m^{\lambda}_{\mu \nu}
\ee
is the number of number of $(d+1)$-term products $(a_1\, b_1) \cdots (a_d\, b_d) g$ that are weakly monotonically increasing, such that $g\in \cyc(\mu)$ and
$(a_1\, b_1) \cdots (a_d\, b_d) g \in \cyc(\nu)$.
These enumerate $d$-step paths in the Cayley graph of $S_n$ from an element in the conjugacy class of cycle type $\mu$ to the class of cycle type $\nu$, that are weakly monotonically increasing in their second elements \cite{GH}.

Equivalently, they are double Hurwitz numbers for $n$-sheeted branched coverings of the Riemann sphere with branch points at $0$ and $\infty$ having ramification profile types $\mu$ and $\nu$,
and an arbitrary number of further branch points, such that the sum of the colengths of their ramification profile lengths is again equal to $d$
\be
\sum_{i=1}^k \ell^*(\mu^{(i)}) = kn - \sum_{i=1}^k \ell(\mu^{(i)}) = d.
\ee
The latter are counted with a sign, which is $ (-1)^{n+d}$ times the parity of the number of branch points \cite{HO2}. The genus is again given by \eqref{riemann_hurwitz_bis}.

\eex

\subsection{Quantum combinatorial Hurwitz numbers}

In this subsection, we introduce three new examples involving weighted enumeration of paths in the Cayley graph of the symmetric group,
which can also be interpreted as weighted enumeration of configurations of branched covers of the Riemann sphere.
Each can be viewed as a $q$-deformation of one the classical examples.
Because of the similarity of the weighting to that of a quantum Bosonic gas, the resulting weighted sums will be identified as {\it quantum} Hurwitz numbers.


\bex{\bf $E(q)$. Quantum Hurwitz numbers (i).}
\label{ex5}
\bea
G(z) =&\&E(q, z) \deq \prod_{k=0}^\infty (1+ q^k z) = \sum_{i=0}^\infty E_i(q) z^i,
\label{GEq} \\
E_i(q)&\& \deq \frac{q^{\frac{1}{2}i(i-1)}}{\prod_{j=1}^i (1-q^j)}, \quad i \ge 1,
\quad E_\lambda(q) = \prod_{i=1}^{\ell(\lambda)}\frac{q^{\frac{1}{2}\lambda_i(\lambda_i -1)}}{\prod_{j=1}^{\lambda_i} (1-q^j)} \\
E(q, \beta, \JJ) &\&= \prod_{a=1}^n \prod_{k=0}^\infty (1+q^k \beta \JJ_a), \\
r^{(E(q), \beta)}_j &\&= \prod_{k=0}^\infty (1+ q^k \beta j), \\
r^{(E(q), \beta)}_\lambda &\&= \prod_{k=0}^\infty \prod_{(i,j)\in \lambda} (1+ q^k \beta (j-i))
= \prod_{k=0}^\infty (zq^k)^{\abs{\lambda}} (1/(zq^k))_\lambda
\\
T^{E(q)}_j (\beta)&\& = - \sum_{i=1}^j\Li_2(q, -\beta i), \quad T^{E(q)}_{-j}(\beta) = \sum_{i=0}^j\Li_2(q, \beta i), \quad j >0.
\label{rEq}
\eea
This weight generating function is related to the quantum dilogarithm function by
\be
E(q, z) = e^{-\left({1\over 1-q}\right)\Li_2(q, -z)}, \quad \Li_2(q, z) \deq (1-q) \sum_{k=1}^\infty \frac{z^k}{k (1- q^k)}.
\ee
The coefficients $E_i(q)$ are themselves generating functions for the number of partitions having exactly $i$ or $i-1$ parts, all distinct.
The coefficient $ F_{E(q)}^d(\mu, \nu)$ is
\be
F_{E(q)}^d(\mu, \nu) = \sum_{\lambda, \ \abs{\lambda}=d} E_\lambda (q) \, m^\lambda_{\mu \nu}
= (d!)^{-1} \sum_{\lambda, \ \abs{\lambda}=d} \tilde{E}_\lambda (q) \, \tilde{m}^\lambda_{\mu \nu}
\label{F_Eqd}
\ee
Its combinatorial interpretation is, as explained in \autoref{weighted_expansions},
as a weighted enumeration of paths in the Cayley graph of $S_n$ from the conjugacy class of type $\mu$ to the class $\nu$, where paths of signature $\lambda$ have weighting factor $E_\lambda(q)$. The geometric interpretation will be detailed in \autoref{hurwitz_Eq}. 
\br
The definition of the quantum dilogarithm is not uniform in the literature. What is referred to in \cite{FK} as the quantum dilogarithm is
\be
\Psi(z) \deq E(q, -z) = e^{-{1\over 1-q}\Li_2(q, z)}.
\ee
The notation $\Li_2(q,z)$ used here is natural since, for $q = e^{-\epsilon}, \abs{q} <1$, the leading term as $\epsilon \ra 0$ coincides with the classical dilogarithm in the rescaled argument $\frac{z}{\epsilon}$:
\be
\Li_2(q,z) \sim \sum_{m=1}^\infty \frac{(\frac{z}{\epsilon})^m}{m^2} = \Li_2\left(\frac{z}{\epsilon}\right).
\ee
\er
A slight modification of this is obtained by removing the $q^0$ term in the product,
giving the weight generating function
\be
E'(q, z)\deq \prod_{k=1}^\infty (1+ q^k z).
\label{GE'q}
\ee
The relevant coefficients in the weighted expansions are thus modified to
\bea
E'_i(q)&\& \deq \frac{q^{\frac{1}{2}i(i+1)}}{\prod_{j=1}^i (1-q^j)}, \quad i \ge 1,
\quad E'_\lambda(q) = \prod_{i=1}^{\ell(\lambda)}\frac{q^{\frac{1}{2}\lambda_i(\lambda_i +1)}}{\prod_{j=1}^{\lambda_i} (1-q^j)} \\
E'(q, \JJ) &\&= \prod_{a=1}^n \prod_{k=1}^\infty (1+q^k z\JJ_a), \\
r^{(E'(q), \beta)}_j &\&= \prod_{k=1}^\infty (1+ q^k \beta j), \\
r^{(E'(q), \beta)}_\lambda &\&= \prod_{k=1}^\infty \prod_{(i,j)\in \lambda} (1+ q^k \beta (j-i))
= \prod_{k=1}^\infty (\beta q^k)^{\abs{\lambda}} (1/(\beta q^k))_\lambda
\eea
The coefficient $ F_{E'(q)}^d(\mu, \nu)$ is
\be
F_{E'(q)}^d(\mu, \nu) = \sum_{\lambda, \ \abs{\lambda}=d} E'_\lambda (q) \, m^\lambda_{\mu \nu}
= (d!)^{-1} \sum_{\lambda, \ \abs{\lambda}=d} \tilde{E}'_\lambda (q) \, \tilde{m}^\lambda_{\mu \nu}
\ee
Its combinatorial interpretation  is the same as in the previous case, with the weighting factor 
$E_\lambda(q)$ replaced by $E'_\lambda(q)$. 

It may be viewed as weighted sums over branched covers, in which the weights are closely related to distributions for
Bosonic gases, with the parameter $q$ interpreted as
\be
q = e^{-{ \hbar \omega_0 \over k_B T}},
\ee
for some fundamental frequency $\omega_0$ and linear energy spectrum,
with the energy levels proportional to the total ramification defect.
The case $E'(q)$ is obtained by removal of the zero energy levels, giving a distribution that more closely resembles that of the Bosonic gas. In the classical limit $q\ra 1$, we recover \autoref{ex1}.
For a detailed study of the leading term semiclassical asymptotics, see \cite{HOr}.

\eex
\bex{\bf $H(q)$. Quantum Hurwitz numbers (ii).}
\label{ex6}
\noindent This is the dual $\tilde{E}(q)$ of the weight generating function of \autoref{ex5}.

\bea
\tilde{G}(z) =H(q,z)&\&\deq \prod_{k=0}^\infty (1-q^k z)^{-1} =\tilde{E}(q,z) = \sum_{i=0}^\infty H_i(q)z^i,
\label{GHq} \\
H_i(q) &\&\deq \frac{1}{\prod_{j=1}^i (1-q^j)},
\quad H_\lambda(q) = \prod_{i=1}^{\ell(\lambda)}\frac{1}{\prod_{j=1}^{\lambda_i} (1-q^j)} \\
H(q, \beta, \JJ) &\&= \prod_{k=0}^\infty\prod_{a=1}^n (1- q^k \beta \JJ_a)^{-1}, \\
r^{(H(q), \beta)}_j &\&= \prod_{k=0}^\infty (1- q^k \beta j)^{-1}, \\
r^{(H(q), \beta)}_\lambda &\&= \prod_{k=0}^\infty \prod_{(i,j)\in \lambda} (1- q^k \beta (j-i)) ^{-1}
= \prod_{k=0}^\infty (-1/(\beta q^k))^{-\abs{\lambda}} (-1/(\beta q^k))^{-1}_\lambda. \\
T^{H(q)}_j(\beta) &\& = \sum_{i=1}^j\Li_2(q, \beta i), \quad T^{H(q)}_{-j}(\beta) = -\sum_{i=1}^{j-1}\Li_2(q, - \beta i), \quad j >0.
\label{rHq}
\eea

The coefficients $H_i(q)$ of the weight generating function in this case are generating functions for the number of partitions having at most $i$ parts, which need not be distinct. The modification corresponding to removing the zero energy level state is based similarly on the generating function
\be
H'(q, z) \deq \prod_{k=1}^\infty (1- q^k z)^{-1}.
\label{GH'q}
\ee

The coefficient $ F_{H}^d(\mu, \nu)$ is
\be
F_{H(q)}^d(\mu, \nu) = \sum_{\lambda, \ \abs{\lambda}=d} H_\lambda (q)\, m^\lambda_{\mu \nu}
= (d!)^{-1} \sum_{\lambda, \ \abs{\lambda}=d} \tilde{H}_\lambda (q)\, \tilde{m}^\lambda_{\mu \nu}
\label{F_Hqd}
\ee
Its combinatorial interpretation  is the same as in the previous cases,  as the weighted enumeration of paths in the Cayley graph of $S_n$ from the conjugacy class of type $\mu$ to the class $\nu$ where paths of signature $\lambda$ have weighting factor $H_\lambda(q)$.

The geometric interpretation is detailed in \autoref{hurwitz_Hq}. It may be viewed as a signed version of the weighted Hurwitz numbers associated to the Bose gas interpretation, with the sign again determined by the parity of the number of branch points. In the classical limit $q\ra 1$, we again recover \autoref{ex1}.

\eex


\bex{\bf $Q(q,p)$. Double quantum Hurwitz numbers.}
\label{ex7}
\noindent This is the composite of examples \ref{ex5} and \ref{ex6}, with the weight generating function
  formed from their product, with two different quantum deformation parameters $q$ and $p$.
\bea
G(z) = Q(q, p, z) &\& \deq E(q,z)H(p,z) = \prod_{k=0}^\infty (1+ q^k z) (1- p^k z)^{-1} = \sum_{i=0}^\infty Q_i(q,p)z^i, \\
Q_i(q,p) &\& \deq \sum_{m=0}^i q^{\frac{1}{2}m(m-1)} \left(\prod_{j=1}^m (1-q^j) \prod_{j=1}^{i -m}(1-p^j)\right)^{-1},
\quad Q_\lambda(q,p) =\prod_{i=1}^{\ell(\lambda)} Q_{\lambda_i}(q,p), \cr
&\& \\
Q(q, p, \beta, \JJ) &\&= \ E(q,\beta, \JJ) H(p, \beta, \JJ), \\
r^{Q(q, p), \beta)}_j &\&= \prod_{k=0}^\infty \frac{1+ q^k \beta j}{1- p^k \beta j}, \\
r^{Q(q, p), \beta)}_\lambda&\&
= \prod_{k=0}^\infty \prod_{(i,j)\in \lambda} \ \frac{1+ q^k \beta (j-i)}{1-p^k \beta (j-i)}
= \prod_{k=0}^\infty (-q/p)^{k\abs{\lambda}} \frac{(1/(\beta q^k))_\lambda}{(-1/(\beta p^k))_\lambda}, \\
T^{Q(q, p)}_j(\beta) &\& = \sum_{i=1}^j\Li_2(p, \beta i) - \sum_{i=1}^j\Li_2(q, -\beta i), \cr
T^{Q(q, p)}_{-j} (\beta)&\& = -\sum_{i=1}^{j-1}\Li_2(p, - \beta i) +\sum_{i=1}^{j-1}\Li_2(q, \beta i), \quad j >0.
\label{rQqp}
\eea

The coefficient $ F_{Q(q,p)}^d(\mu, \nu)$ is
\be
F_{Q(q, p)}^d(\mu, \nu) = \sum_{\lambda, \ \abs{\lambda}=d} Q_\lambda (q,p)\, m^\lambda_{\mu \nu}
= (d!)^{-1} \sum_{\lambda, \ \abs{\lambda}=d} \tilde{Q}_\lambda (q,p)\, \tilde{m}^\lambda_{\mu \nu}.
\label{F_Qqpd}
\ee

Geometrically, these are interpreted in  \autoref{hurwitz_Qqp}. as the composite of two types of weighted enumerations; i.e., two species of branch points, one of which is counted with the weight corresponding to a Bosonic gas as in \autoref{ex5}, the other counted, as in \autoref{ex6}, with signs determined by the parity of the number of such branch points. In the classical limit $q\ra 1, p \ra 1$, we again recover \autoref{ex1}.

\eex

\subsection{Generating functions for Hurwitz numbers: classical counting of branched covers}

For \autoref{ex1}, the generating $\tau$-function is 
\bea
\tau^{(\exp, \beta)} ({\bf t}, {\bf s}) &\& = \sum_{\lambda}
e^{\frac{\beta}{2} \sum_{i=1}^{\ell(\lambda)}\lambda_i (\lambda_i - 2i +1)}
s_\lambda({\bf t}) s_\mu({\bf s}) \cr
&\& = \sum_{d=0}^\infty \ \sum_{\mu, \nu, \, \abs{\mu}=\abs{\nu}} H_{\exp}^d(\mu, \nu) \beta^d p_\mu ({\bf t}) p_\nu({\bf s}),
\eea
where
\be
 H^d_{\exp} (\mu, \nu)
= \frac{1}{d!} H(\mu^{(1)} = (2, 1^{n-2}), \dots, \mu^{(d)} =  (2, 1^{n-2}), \mu, \nu)
\ee
is $\frac{1}{d!}$ times Okounkov's double Hurwitz number ${\rm Cov}_d(\mu, \nu)$ \cite{Ok}; that is, the number of $n=\abs{\mu}=\abs{\nu}$ sheeted branched covers with branch points of ramification type $\mu$ and $\nu$ at the points $0$ and $\infty$,
and $d$ further simple branch points.

For \autoref{ex2}, the generating $\tau$-function is \cite{GH, HO2}
\bea
\tau^{(E, \beta} ({\bf t}, {\bf s}) &\& = \sum_{\lambda}
z^{\abs{\lambda}} (\beta^{-1})_\lambda s_\lambda({\bf t}) s_\mu({\bf s}) \cr
&\& = \sum_{d=0}^\infty \beta^d \sum_{\mu, \nu,\; \abs{\mu}=\abs{\nu}} H_E^d(\mu, \nu) p_\mu ({\bf t}) p_\nu({\bf s}),
\eea
where
\be
H_E^d(\mu, \nu) = \sum_{\mu^{(1)}, \ \ell^*(\mu_1) =d}
H(\mu^{(1)}, \mu, \nu)
\ee
is now interpreted as the number of $n=\abs{\mu}=\abs{\nu}=\abs{\mu^{(1)}}$ sheeted branched covers with branch points of ramification type $\mu$ and $\nu$ at $0$ and $\infty$,
and one further branch point, with colength $\ell^*(\mu^{(1)}) =d$; i.e., the case of Belyi curves \cite{Z, KZ, HO2, AC1} or {\em dessins d'enfants}.

For \autoref{ex3}, the generating $\tau$-function is \cite{HO2} is
\bea
\tau^{(E^k , \beta)} ({\bf t}, {\bf s}) &\& = \sum_{\lambda}
\beta^{\abs{\lambda}} (1/\beta)_\lambda s_\lambda({\bf t}) s_\mu({\bf s}) \cr
&\& = \sum_{d=0}^\infty \beta^d \sum_{\mu, \nu, \ \abs{\mu}=\abs{\nu} =n} H_{E^k}^d(\mu, \nu) p_\mu ({\bf t}) p_\nu({\bf s}),
\eea
where
\be
H_{E^k}^d(\mu, \nu) = \sum_{\substack{\mu^{(1)}, \dots, \mu^{(k)} \\ \sum_{i=1}^k\ell^*(\mu_i) =d}}
H(\mu^{(1)}, \dots \mu^{(k)}, \mu, \nu)
\ee
is now interpreted \cite{HO2} as the number of $n=\abs{\mu}=\abs{\nu}=\abs{\mu^{(i)}}$ sheeted branched covers with branch points of ramification type $\mu$ and $\nu$ at $0$ and $\infty$,
and (at most) $k$ further branch points, the sums of whose colengths is $d$.

For \autoref{ex4}, the generating $\tau$-function is \cite{GH, HO2}
\bea
\tau^{(H, \beta)} ({\bf t}, {\bf s}) &\& = \sum_{\lambda}
(-\beta)^{\abs{\lambda}} \left(-\beta^{-1}\right)_{\lambda}
s_\lambda({\bf t}) s_\mu({\bf s}) \cr
&\& = \sum_{d=0}^\infty \beta^d\sum_{\mu, \nu,\; \abs{\mu}=\abs{\nu}} H^d_{H}(\mu, \nu) p_\mu ({\bf t}) p_\nu({\bf s})
\eea
where
\be
H^d_{H}(\mu, \nu) = (-1)^{n+d}\sum_{j=1}^\infty (-1)^j \sum_{\substack{\mu^{(1)},\dots,\mu^{(j)} \\ \sum_{i=1}^j\ell^*(\mu_i) = d}}
H(\mu^{(1)}, \dots \mu^{(j)}, \mu, \nu)
\ee
is now interpreted as the signed counting of $n =\abs{\mu}=\abs{\nu}$ sheeted branched covers with branch points of ramification type $\mu$ and $\nu$ at $0$ and $\infty$,
and any number further branch points, the sum of whose colengths is $d $,
with sign determined by the parity of the number of branch points \cite{HO2}.

For each of the four cases $G= \exp, E, E^k$ and $H$, we have thus shown the equality
\be
F^d_G(\mu, \nu) = H^d_G(\mu, \nu)
\ee
between the combinatorial weighted path enumeration and the weighted (signed) branched covering enumeration.

\subsection{The $\tau$-functions $\tau^{(E(q), \beta)}$, $\tau^{(H(q), \beta)}$ and $\tau^{(Q(q,p), \beta)}$ as generating functions for enumeration of quantum paths}

The particular cases
\bea
\tau^{(E(q), \beta)} ({\bf t}, {\bf s}) &\&\deq \sum_{\lambda} r^{(E(q), \beta)}_\lambda s_\lambda({\bf t}) s_\lambda({\bf s})
\label{tau_Eq} \\
&\&= \sum_{d=0}^\infty \sum_{\substack{\mu, \nu \\ \abs{\mu} = \abs{\nu}}}\beta^d F_{E(q)}^d(\mu, \nu) p_\mu({\bf t}) p_\nu({\bf s}),
\label{F_Eq} \\
\tau^{(H(q), \beta)} ({\bf t}, {\bf s}) &\&\deq \sum_{\lambda} r^{(H(q), \beta)}_\lambda  s_\lambda({\bf t}) s_\lambda({\bf s})
\label{tau_Hq}\\
&\&= \sum_{d=0}^\infty \sum_{\substack{\mu, \nu \\ \abs{\mu} = \abs{\nu}}}z^d F_{H(q)}^d(\mu, \nu) p_\mu({\bf t}) p_\nu({\bf s}),
\label{F_Hq} \\
\tau^{(Q(q, p), \beta)} ({\bf t}, {\bf s}) &\&\deq \sum_{\lambda} r^{(Q(q, p), \beta)}_\lambda s_\lambda({\bf t}) s_\lambda({\bf s})
\label{tau_Qqp} \\
&\&= \sum_{d=0}^\infty \sum_{\substack{\mu, \nu \\ \abs{\mu} = \abs{\nu}}}\beta^d F_{Q(q,p)}^d(\mu, \nu) p_\mu({\bf t}) p_\nu({\bf s}).
\label{F_Qqp}
\eea
may be viewed as special $q$-deformations of the generating functions associated to examples
\autoref{ex1}, \autoref{ex2}, with $G(z) =1+z$, $G(z)=(1-z)^{-1}$ respectively, and the hybrid combination generated by the ratio $\frac{1+z}{1-z}$. The former were considered previously in \cite{GH, HO2}, and given both combinatorial and geometric interpretations in terms of weakly or strictly monotonic paths in the Cayley graph.

\br
Note that, for the special values of the flow parameters $({\bf t}, {\bf s})$ given by trace invariants of a pair of commuting $M\times M$ matrices, $X$ and $Y$,
\be
t_i =\frac{1}{i} \tr(X^i), \quad s_i =\frac{1}{i} \tr(Y^i),
\ee
with eigenvalues $(x_1, \dots, x_M)$, $(y_1, \dots, y_M)$,
these may be viewed as special cases of the two types of basic hypergeometric functions of matrix arguments \cite{GR, OrS}.
\er

\subsection{Classical limits of examples \texorpdfstring{$E(q)$, $H(q)$ and $Q(q,p)$}{E(q), H(q) and Q(q,p)}}
Setting $q= e^\epsilon$ for some small parameter, and taking the leading term contribution in the limit $\epsilon \ra 0$, we obtain
\be
\lim_{\epsilon \ra 0} E(q, \epsilon z) = e^z
\ee
and therefore, taking the scaled limit with $ z \ra \epsilon z$, we obtain
\be
\lim_{\epsilon \ra 0} \tau^{(E(q), \epsilon \beta)} ({\bf t}, {\bf s}) = \tau^{(\exp, \beta)} ({\bf t}, {\bf s})
\ee
Similarly, we have
\be
\lim_{\epsilon \ra 0} H(q, \epsilon z) = e^z
\ee
and hence
\be
\lim_{\epsilon \ra 0} \tau^{(H(q), \epsilon \beta)} ({\bf t}, {\bf s}) = \tau^{(\exp, \beta)} ({\bf t}, {\bf s}).
\ee
And finally, for the double quantum Hurwitz case, \autoref{ex7}, setting
\be
q= e^\epsilon, \quad p = e^{\epsilon'}
\ee
and replacing $z$ by $z(\frac{1}{\epsilon} + \frac{1}{\epsilon'})$, we get
\be
\lim_{\epsilon, \epsilon' \ra 0} Q(q, p, \frac{z\epsilon \epsilon'}{\epsilon + \epsilon'}) = e^z
\ee
and hence
\be
\lim_{\epsilon, \epsilon' \ra 0} \tau^{(Q(q,p), \frac{\beta\epsilon \epsilon'}{\epsilon + \epsilon'})} ({\bf t}, {\bf s})
= \tau^{(\exp, \beta)} ({\bf t}, {\bf s}).
\ee
Thus, we recover Okounkov's classical double Hurwitz number generating function $\tau^{\exp(z)} ({\bf t}, {\bf s})$ as the classical limit in each case.

\section{Quantum Hurwitz numbers}
\label{quantum_hurwitz}

We proceed to the interpretation of the quantities $H_{E(q)}^d(\mu, \nu)$,
$H_{H(q)}^d(\mu, \nu)$ and $H_{Q(q,p)}^d(\mu, \nu)$ as quantum weighted
enumerations of branched coverings of the Riemann sphere. 

\subsection{Symmetrized monomial sums and \texorpdfstring{$q$}{q}-weighted Hurwitz sums}

We begin by recalling three symmetrized monomial summation formulae that will be needed in what follows.
These are easily proved (e.g., by recursive diagonal summation of the geometric series involved).

Let $\overline{\Cb[x_1, \dots, x_k]}$ be the completion of the field extension of
$\Cb$ by $k$ indeterminates, viewed as a normed vector space with norm $\abs{\cdots}$ and $0 < \abs{x_i} < 1$,
so that the corresponding geometric series converge
\be
\sum_{m=0}^\infty x_i^m = \frac{1}{1- x_i}, \quad i=1, \dots, k.
\ee
Then
\bea
\sum_{\sigma\in S_k} \sum_{0 \le i_1 < \cdots < i_k}^\infty x_{\sigma(1)}^{i_1} \cdots x_{\sigma(k)}^{i_k}
&\&= \sum_{\sigma\in S_k} \frac{x^{k-1}_{\sigma(1)} x^{k-2}_{\sigma(2)}\cdots x_{\sigma(k-1)}}{
(1- x_{\sigma(1)}) (1- x_{\sigma(1)} x_{\sigma(2)}) \cdots (1- x_{\sigma(1)} \cdots x_{\sigma(k)})}
\label{strictly_ordered_monomial sums_0} \\
\sum_{\sigma\in S_k} \sum_{1 \le i_1 < \cdots < i_k}^\infty x_{\sigma(1)}^{i_1} \cdots x_{\sigma(k)}^{i_k}
&\&= \sum_{\sigma\in S_k} \frac{x^{k}_{\sigma(1)} x^{k-1}_{\sigma(2)} \cdots x_{\sigma(k)}}{
(1- x_{\sigma(1)}) (1- x_{\sigma(1)} x_{\sigma(2)}) \cdots (1- x_{\sigma(1)} \cdots x_{\sigma(k)})}
\label{strictly_ordered_monomial sums_1} \\
\sum_{\sigma\in S_k} \sum_{0 \le i_1 \le \cdots \le i_k}^\infty x_{\sigma(1)}^{i_1} \cdots x_{\sigma(k)}^{i_k}
&\&= \sum_{\sigma\in S_k} \frac{1}{(1- x_{\sigma(1)}) (1- x_{\sigma(1)} x_{\sigma(2)})
\cdots (1- x_{\sigma(1)} \cdots x_{\sigma(k)})}
\label{weakly_ordered_monomial sums_0}
\eea

In what follows, we let $(\mu^{(1)}, \dots, \mu^{(k)})$ denote a set of partitions of weight $\abs{\mu^{(i)}}=n$, and choose the $x_i$'s to be
\be
x_i \deq q^{\ell^*(\mu^{(i)})}.
\ee
For the generating functions $E(q), E'(q)$ and $H(q) $, using eqs.~\eqref{strictly_ordered_monomial sums_0}, 
\eqref{strictly_ordered_monomial sums_1} and \eqref{weakly_ordered_monomial sums_0} we have the following weighting factors
\bea
W_{E(q)} (\mu^{(1)}, \dots, \mu^{(k)}) &\& \deq {1\over |\aut(\lambda)|}
\sum_{\sigma\in S_k} \sum_{0 \le i_1 < \cdots < i_k}^\infty q^{i_1 \ell^*(\mu^{(\sigma(1))})} \cdots q^{i_k \ell^*(\mu^{(\sigma(k))})} \cr
&\& = {1\over  |\aut(\lambda)|}\sum_{\sigma\in S_k} \frac{q^{(k-1) \ell^*(\mu^{(\sigma(1))})} \cdots q^{\ell^*(\mu^{(\sigma(k-1))})}}{
(1- q^{\ell^*(\mu^{(\sigma(1))})}) \cdots (1- q^{\ell^*(\mu^{(\sigma(1))})} \cdots q^{\ell^*(\mu^{(\sigma(k))})})}, \cr
&\&
\label{W_E_q}
\eea
\bea
W_{E'(q)} (\mu^{(1)}, \dots, \mu^{(k)}) &\& \deq
 {1\over |\aut(\lambda)|} \sum_{\sigma\in S_k} \sum_{1 \le i_1 < \cdots < i_k}^\infty q^{i_1 \ell^*(\mu^{(\sigma(1))})} \cdots q^{i_k \ell^*(\mu^{(\sigma(k))})} \cr
&\&=  {1\over  |\aut(\lambda)|} \sum_{\sigma\in S_k} \frac{q^{k \ell^*(\mu^{(\sigma(1))})} \cdots q^{\ell^*(\mu^{(\sigma(k))})}}{
(1- q^{\ell^*(\mu^{(\sigma(1))})}) \cdots (1- q^{\ell^*(\mu^{(\sigma(1))}} \cdots q^{\ell^*(\mu^{(\sigma(k))})})} \cr
&\& =  {1\over  |\aut(\lambda)|} \sum_{\sigma\in S_k} \frac{1}{
(q^{-\ell^*(\mu^{(\sigma(1))})} -1) \cdots (q^{-\ell^*(\mu^{(\sigma(1))})} \cdots q^{-\ell^*(\mu^{(\sigma(k))})}-1)}, \cr
&\&
\label{W_Eprime_q}
\eea
\bea
W_{H(q)} (\mu^{(1)}, \dots, \mu^{(k)}) &\& \deq
{(-1)^{\ell^*(\lambda)}\over   |\aut(\lambda)|}\sum_{\sigma\in S_k} \sum_{0 \le i_1 \le \cdots \le i_k}^\infty q^{i_1 \ell^*(\mu^{(\sigma(1))})} \cdots q^{i_k \ell^*(\mu^{(\sigma(k))})} \cr
&\&= {(-1)^{\ell^*(\lambda)}  \over  |\aut(\lambda)|}\sum_{\sigma\in S_k} \frac{1}{
(1- q^{\ell^*(\mu^{(\sigma(1))})}) \cdots (1- q^{\ell^*(\mu^{(\sigma(1))})} \cdots q^{\ell^*(\mu^{(\sigma(k))})})}, \cr
&\&
\label{W_H_q}
\eea
where $\lambda$ is the partition with parts $(\ell^*(\mu^{(1)}), \dots, \ell^*(\mu^{(k)}))$.


\subsection{Quantum Hurwitz numbers: the case \texorpdfstring{$E(q)$}{E(q)}}
\label{E(q)}

Substituting eq.~\eqref{W_E_q} into \eqref{tau_Eq}, \eqref{tau_G_tilde_H} and \eqref{Hd_G} gives
\begin{theorem}
\label{hurwitz_Eq}
\be
\tau^{(E(q), \beta)} ({\bf t}, {\bf s}) =
\sum_{d=0}^\infty \beta^d \sum_{\substack{\mu, \nu \\ \abs{\mu}=\abs{\nu}}}
H^d_{E(q)}(\mu, \nu) p_\mu({\bf t}) p_\nu({\bf s}),
\ee
where
\be
H^d_{E(q)}(\mu, \nu) \deq \sum_{k=0}^\infty \sideset{}{'}\sum_{\substack{\mu^{(1)}, \dots \mu^{(k)} \\ \sum_{i=1}^k \ell^*(\mu^{(i)})= d}}
W_{E(q)}(\mu^{(1)}, \dots, \mu^{(k)}) H(\mu^{(1)}, \dots, \mu^{(k)}, \mu, \nu)
\label{Hd_Eq}
\ee
are the weighted (quantum) Hurwitz numbers that count the number of branched coverings with genus $g$ given by \eqref{riemann_hurwitz_bis}
with weight $W_{E(q)}(\mu^{(1)}, \dots, \mu^{(k)})$ for every branched covering of type $ (\mu^{(1)}, \dots, \mu^{(k)}, \mu, \nu)$.
\end{theorem}
From eq.~\eqref{F_Eq} follows:
\begin{corollary}
\label{F_equals_H_Eq}
The weighted (quantum) Hurwitz number for the branched coverings of the Riemann sphere with genus given by
\eqref{riemann_hurwitz_bis} is equal to the combinatorial Hurwitz number given by formula \eqref{F_Eqd}
enumerating weighted paths in the Cayley graph:
\be
H^d_{E(q)}(\mu, \nu) = F^d_{E(q)}(\mu, \nu).
\ee
\end{corollary}

\subsection{Dual quantum Hurwitz numbers: the case \texorpdfstring{$H(q)$}{H(q)}}
\label{H(q)}

We proceed similarly for this case. Substituting eq.~\eqref{W_H_q} into \eqref{tau_Hq}, \eqref{tau_G_tilde_H} and \eqref{Hd_G} gives
\begin{theorem}
\label{hurwitz_Hq}
\be
\tau^{(H(q), \beta)} ({\bf t}, {\bf s}) =
\sum_{d=0}^\infty \beta^d \sum_{\substack{\mu, \nu \\ \abs{\mu}= \abs{\nu}}}
H^d_{H(q)}(\mu, \nu) p_\mu({\bf t}) p_\nu({\bf s}),
\ee
where
\be
H^d_{H(q)}(\mu, \nu) \deq \sum_{k=0}^\infty \sideset{}{'}\sum_{\substack{\mu^{(1)}, \dots \mu^{(k)} \\ \sum_{i=1}^d \ell^*(\mu^{(i)})= d}}
W_{H(q)}(\mu^{(1)}, \dots, \mu^{(k)}) H(\mu^{(1)}, \dots, \mu^{(k)}, \mu, \nu)
\label{Hd_Hq}
\ee
are the weighted, signed (quantum) Hurwitz numbers that count the number of branched coverings with genus $g$ given by \eqref{riemann_hurwitz_bis} and sum of colengths $k$, with weight $W_{H(q)}(\mu^{(1)}, \dots, \mu^{(k)})$ for every branched covering of type $ (\mu^{(1)}, \dots, \mu^{(k)}, \mu, \nu)$.
\end{theorem}

From eq.~\eqref{F_Hq} follows:
\begin{corollary}
\label{F_equals_H_Hq}
The weighted (quantum) Hurwitz number for the branched coverings of the Riemann sphere with genus given by
\eqref{riemann_hurwitz_bis} is again equal to the combinatorial Hurwitz number given by formula \eqref{F_Hqd}
enumerating weighted paths in the Cayley graph:
\be
H^d_{H(q)}(\mu, \nu) = F^d_{H(q)}(\mu, \nu).
\ee
\end{corollary}

\subsection {Double quantum Hurwitz numbers: the case \texorpdfstring{$Q(q,p)$}{Q(q,p)}}

This case can be understood by combining the results for the previous two multiplicatively.
Since
\be
r^{(Q(q,p),\beta)}_\lambda = r^{(E(q), \beta)}_\lambda\, r^{(H(p), \beta)}_\lambda,
\ee
it follows that:
\begin{theorem}
\label{hurwitz_Qqp}
\be
\tau^{(Q(q, p), \beta)} ({\bf t}, {\bf s}) = \sum_{d=0}^\infty \beta^d \sum_{\substack{\mu, \nu \\ \abs{\mu}=\abs{\nu}}}
H^d_{Q(q,p)}(\mu, \nu) p_\mu({\bf t}) p_\nu({\bf s}),
\ee
where
\bea
H^d_{Q(q,p)}(\mu, \nu) \deq &\& \sum_{k=0}^\infty \sum_{m=0}^\infty (-1)^{m}
\quad \mathclap{\sideset{}{'}\sum_{\substack{\mu^{(1)}, \dots \mu^{(k)}, \nu^{(1)}, \dots \nu^{(m)} \\ \sum_{i=1}^k \ell^*(\mu^{(i)}) +\sum_{i=1}^m \ell^*(\nu^{(i)})= d}}}
\qquad \qquad W_{E(q)}(\mu^{(1)}, \dots, \mu^{(k)}) \cr
&\&
\times W_{H(p)}(\nu^{(1)}, \dots, \nu^{(m)})H(\mu^{(1)}, \dots, \mu^{(k)}, \nu^{(1)}, \dots, \nu^{(m)}, \mu, \nu)
\label{Hd_Qqp}
\eea
are the weighted (quantum) Hurwitz numbers that count the number of branched coverings with genus $g$ given by \eqref{riemann_hurwitz_bis} and sum of colengths $d$, with two mutually independent species of branch points, the first $(\mu^{(1)}, \dots, \mu^{(k)})$ having weight
$W_{E(q)}(\mu^{(1)}, \dots, \mu^{(k)})$, the second $(\nu^{(1)}, \dots, \nu^{(m)})$, signed weight
$ (-1)^{m}W_{H(q)}(\nu^{(1)}, \dots, \nu^{(m)})$ for every branched covering of type $ (\mu^{(1)}, \dots, \mu^{(k)}, \nu^{(1)}, \dots, \nu^{(m)}, \mu, \nu)$.\end{theorem}

It follows from eq.~\eqref{F_Eq} that:
\begin{corollary}
\label{F_equals_H_Qqp}
The weighted (quantum) Hurwitz number for the branched coverings of the Riemann sphere with genus given by
\eqref{riemann_hurwitz_bis} is again equal to the combinatorial Hurwitz number given by formula \eqref{F_Qqpd}
enumerating weighted paths in the Cayley graph:
\be
H^d_{Q(q,p)}(\mu, \nu) = F^d_{Q(q,p)}(\mu, \nu).
\ee
\end{corollary}

\subsection{Bose gas model}

A slight modification of \autoref{ex5} consists of replacing the generating function $E(q,z)$ defined in eq.~\eqref{GEq}
by $E'(q,z)$, as defined in \eqref{GE'q} and \eqref{GH'q}. The effect of this is simply to replace the weighting factors
$ \frac{1}{1- q^{\ell^*(\mu)}} $ in eq.~\eqref{Hd_Eq} by $ \frac{1}{q^{-\ell^*(\mu)}-1}$.

If we identify
\be
q \deq e^{-{\hbar \omega_0 \over k_B T}}, 
\ee
where $\omega_0$ is the lowest frequency excitation in a gas of identical Bosonic particles and assume the energy spectrum of the particles consists of integer multiples of $\hbar \omega_0$
\be
\epsilon_k = k \hbar \omega_0,
\ee
the relative probability of occupying the energy level $\epsilon_k$ is
\be
\frac{q^k}{1- q^{k}} = \frac{1}{e^{\beta\epsilon_k - 1}},
\ee
which is the energy distribution of a Bosonic gas with vanishing fugacity. If we assign the energy
\be
\epsilon(\mu) \deq \epsilon_{\ell^*(\mu)} = \hbar \ell^*(\mu)\omega_0
\ee
and assign a weight to a configuration $(\mu^{(1)}, \dots, \mu^{(k)}, \mu, \nu)$ that corresponds to the Bosonic gas weight for a state with total energy that of the additional $k$ branch points

\be
\epsilon(\mu^{(1)}, \dots, \mu^{(k)}) = \sum_{i=1}^k \epsilon(\mu^{(i)})
\ee
we obtain the weight
\be
W(\mu^{(1)}, \dots, \mu^{(k)})= \frac{1}{e^{\beta \epsilon(\mu^{(1)}, \dots, \mu^{(k)})}-1}.
\ee
From eq.~\eqref{W_Eprime_q}, the weighting factor $W_{E'(q)} (\mu^{(1)}, \dots, \mu^{(k)})$ for $k$ branch points with ramification profiles $(\mu^{(1)}, \dots, \mu^{(k)})$ is thus the symmetrized product.
\be
W_{E'(q)} (\mu^{(1)}, \dots, \mu^{(k)}) = {1\over \abs{\aut(\lambda)}}\sum_{\sigma\in S_k}
W(\mu^{(\sigma(1)}) \cdots W(\mu^{(\sigma(1)}, \ldots, \mu^{(\sigma(k)}).
\label{W_bosonic_gas_weight}
\ee
If we associate the branch points to the states of the gas and view the
Hurwitz numbers $H(\mu^{(1)}, \dots \mu^{(k)}, \mu, \nu)$ as random variables, the weighted Hurwitz numbers are given, as in eq.~\eqref{Hd_Eq}, by
\be
H^d_{E'(q)}(\mu, \nu) \deq
\sum_{k=0}^\infty \sum_{\substack{\mu^{(1)}, \dots \mu^{(k)} \\ \sum_{i=1}^k \ell^*(\mu^{(i)})= d}}
W_{E'(q)} (\mu^{(1)}, \dots, \mu^{(k)}) H(\mu^{(1)}, \dots, \mu^{(k)}, \mu, \nu)
\label{Hd_E'q}
\ee
Normalizing by the canonical partition function for fixed total energy $d\hbar\omega$,
\be
Z^d_{E'(q)} \deq \sum_{k =0}^\infty \sum_{\substack{\mu^{(1)}, \dots \mu^{(k)} \\ \sum_{i=1}^k \ell^*(\mu^{(i)})= d}}
W_{E'(q)} (\mu^{(1)}, \dots, \mu^{(k)}),
\label{Zk_E'q}
\ee
we may therefore interpret this as an expectation value of the Hurwitz numbers associated to the Bose gas
\be
\expectation{H^d_{E'(q)}(\mu, \nu)} = \frac{H^d_{E'(q)}(\mu, \nu)}{Z^d_{E'(q)}},
\ee
and view the corresponding $\tau$-function
\be
\frac{\tau^{(E'(q), \beta)} ({\bf t}, {\bf s})}{Z^d_{E'(q)}} =
\sum_{d=0}^\infty \beta^d \sum_{\substack{\mu, \nu \\ \abs{\mu}= \abs{\nu}}}
\expectation{H^d_{E'(q)}(\mu, \nu)} p_\mu({\bf t}) p_\nu({\bf s}),
\ee
as a generating function for these expectation values.

\subsection{Formulae for \texorpdfstring{$H^{n-1}_{E(q)}(n)$, $H^{n-1}_{H(q)}(n)$, $H^{n-1}_{E'(q)}(n)$}{H(n) for E(q), H(q), E'(q)} for \texorpdfstring{$n=2, 3, 4, 5$}{n = 2, 3, 4, 5}}
\label{quantum_formulae_low_dim}
The computation of the quantum double Hurwitz numbers $H^d_{E(q)}(\mu, \nu)$,  $H^d_{H(q)}(\mu, \nu)$
and  $H^d_{E'(q)}(\mu, \nu)$ follows from the $S_n$  character tables. Here, we give  explicit expressions
for  the cases $n=2,3, 4$ and $5$  when  $d=n-1$;  i.e. for branched coverings of genus $0$, in the case when 
$\mu = 1^n$, so this is a point with no ramification, while $\nu= (n)$,  corresponding to a branch point 
with maximal ramification.  We denote  the three quantum Hurwitz numbers $H_G^{(n-1)}((1)^n, n)$
for $G = E(q), H(q)$ and $E'(q)$ as $H_G^{(n-1)}( n)$, since there is only one fixed branched point 
of profile type $\mu = (n)$.
Equivalently, these enumerate  the quantum weighted  factorizations
of a full cyclic element $h\in \cyc(n)$ as a product of $n-1$ transpositions. The character tables may
be found, e.g., in \cite{FH};  we just list the relevant  Hurwitz numbers $H(\nu^{(1)}, \dots, \nu^{(k)})$
that contribute to the sums in \eqref{Hd_Eq}, \eqref{Hd_Hq} and \eqref{Hd_E'q}

For $n=2$, the only term contributing  is:
\be
H((2),( 2)) = \frac{1}{2}.
\ee

For $n=3$, the  terms contributing are:
\be
H((3),( 3)) = \frac{1}{3}, \quad H((21), (21), (3)) = 1.
\ee

For $n=4$, the  terms contributing are:
\bea
H((4),( 4)) &\& = \frac{1}{4}, \quad H((31), (2 1^2), (4)) = 1, \cr
&\& \cr
 H((22), (2 1^2), (4)) &\& = \frac{1}{2},
\quad H((21^2), (21^2), (21^2), (4)) = 4.
\eea

For $n=5$, the  terms contributing are:
\bea
H((5),( 5)) &\& = \frac{1}{5}, \quad H((3 1^2), (3 1^2), (5)) = 1, \quad 
 H((3 1^2), (2 1^3), (2 1^3), (5)) = 5,
  \cr
 &\&
   \cr
H((3 2) (2 1^3), (5)) &\& = 1,  \quad  H((4 1) (21^3), (5)) = 1, 
\quad  H((2^2 1), (2  1^3), (2 1^3), (5)) = 5, \cr
 &\& 
  \cr
 H((2 1^3),  (2 1^3), (2 1^3), (2 1^3), (5)) &\& = 25, \quad H((2^2 1), (2^2 1), (5)) = 1, \quad 
 H((2^2 1),  (3 1^2), (5)) = 1. \cr
 &\&
  \eea

From eqs.~\eqref{W_E_q}--\eqref{W_H_q}, \eqref{Hd_Eq}, \eqref{Hd_Hq} and  \eqref{Hd_E'q}, 
we obtain the following expressions for  $ H_{E(q)}^{(n-1)}(n)$, $ H_{H(q)}^{(n-1)}(n)$  
 and $ H_{E'(q)}^{(n-1)}(n)$ for $n=2, 3, 4, 5$: \hfill
\break
\bigskip

\hspace{-2.5cm}
\begin{tabular}{ |p{.3 cm}||p{5.8 cm}|p{5.8 cm}|p{6.2 cm}|  }
 \hline
 \multicolumn{4}{|c|}{Quantum Hurwitz numbers } \\
 \hline
   $n$&  \qquad \qquad $H^{n-1}_{E(q)}(n)$ & \qquad \qquad $H^{n-1}_{H(q)}(n)$ &\qquad \qquad $H^{n-1}_{E'(q)}(n)$ \\
  \hline
      &    &  &   \\
$2$ &\qquad \qquad  \quad $\frac{1}{2(1-q)}$ &\qquad   \qquad  \ $ \frac{1}{2(1-q)} $&
   \qquad\qquad  \quad$\frac{q}{2(1-q)}$\\
     &    &  &   \\
$3$  & \qquad \qquad  $\frac{1+5q}{3(1-q)(1-q^2)}$   & \qquad \quad \  $\frac{5+q}{3(1-q)(1-q^2)}$
 &  \qquad \qquad  $\frac{q^2+5q^3}{3(1-q)(1-q^2)}$\\
     &    &  &   \\
  $4$    & \qquad \quad  $\frac{1+11q+11q^2+q^3}{4(1-q)(1-q^2)(1-q^3)} $ 
     &  \qquad \  $\frac{1+11q+11q^2+q^3}{4(1-q)(1-q^2)(1-q^3)}$
  &\qquad \quad  
$\frac{q^3+11q^4+11q^5+q^6}{4(1-q)(1-q^2)(1-q^3)} $ \\
  &    &  &   \\
$ 5$ &   $\frac{1+19q+39q^2+260q^3+261q^4+241q^5+2179q^6}{5(1-q)(1-q^2)(1-q^3)(1-q^4)}$   
 &   $\frac{2179+241q+261q^2+260q^3+39q^4+19q^5+q^6}{5(1-q)(1-q^2)(1-q^3)(1-q^4)}$ 
  &    $\frac{q^4+19q^5+39q^6+260q^7+261q^8+241q^9+2179q^{10}}{5(1-q)(1-q^2)(1-q^3)(1-q^4)}$   \\
     &    &  &   \\

\hline
 \end{tabular}

\bigskip
\br
By combining the cases  $E(q)$, $H(q)$ multiplicatively, a multiparametric family of generating functions may be obtained, for which the underlying weight generating function is the product
\be
Q({\bf q},  {\bf z}; {\bf p}, {\bf w}) \deq \prod_{i=1}^l E(q_i, w_i) \prod_{j=1}^m H(p_j, z_j).
\ee
The interpretation of such multiparametric  multispecies quantum Hurwitz numbers, both in terms of weighted enumeration of branched covers, and weighted paths in the Cayley graph, is the subject of \cite{H1}.
\er

\br
Since the completion of this work, the above approach to weighted Hurwitz numbers and their generating functions 
has been developed further in \cite {ACEH1}, \cite{ACEH2} so as to include it in the topological recursion program  \cite{EO1, EO2}. The parameter $\beta$ plays a r\^ole analogous to Planck's constant in the asymptotic development of the associated solutions  to the {\em quantum spectral curve equation} as WKB series.
\er

\bigskip
\bigskip 
\noindent
{\small \it Acknowledgements.} The authors would like to thank Marco Bertola, Gaetan Borot,
and Alexander Orlov for helpful discussions. Thanks also to Dennis Stanton for suggesting  comparison 
with the analog \cite{LRS} of Singer cycles in $\GL_n(\mathbb{F}_q)$,   which motivated computing
the explicit formulae of \autoref{quantum_formulae_low_dim}.

Work supported by the Natural Sciences and Engineering Research Council of Canada (NSERC) and the Fonds de recherche du Qu\'ebec -- Nature et technologies (FRQNT).

\bigskip


\newcommand{\arxiv}[1]{\href{http://arxiv.org/abs/#1}{arXiv:{#1}}}

\end{document}